\documentclass[11pt,a4paper]{article}
\usepackage{jheppub}
\usepackage[T1]{fontenc}
\usepackage{lmodern}
\usepackage{amsmath,amssymb,amsthm,mathtools}
\usepackage{booktabs}
\usepackage{array}
\usepackage{microtype}
\usepackage{tikz}
\usetikzlibrary{arrows.meta,positioning}

\makeatletter
\def\@fpheader{\relax}
\makeatother

\hypersetup{
pdftitle={A computer assisted existence proof for a non-Schwarzschild black hole in Einstein--Weyl gravity}, pdfauthor={Kevin Goldstein and Vishnu Jejjala}, pdfsubject={Computer assisted existence proof in Einstein--Weyl gravity}, pdfkeywords={Black Holes, Classical Theories of Gravity} }

\let\OriginalTableOfContents\tableofcontents
\renewcommand{\tableofcontents}{{\small\renewcommand{\baselinestretch}{1}\selectfont\enlargethispage{2\baselineskip}\OriginalTableOfContents}}

\renewcommand{\afterTocSpace}{\smallskip}
\renewcommand{\afterTocRuleSpace}{\clearpage}

\allowdisplaybreaks

\newtheorem{theorem}{Theorem}[section]
\newtheorem{proposition}[theorem]{Proposition}
\newtheorem{lemma}[theorem]{Lemma}
\newtheorem{corollary}[theorem]{Corollary}
\theoremstyle{definition}

\theoremstyle{remark}

\newcommand{\dd}{\mathrm{d}}

\newcommand{\cE}{\mathcal{E}}
\newcommand{\cN}{\mathcal{N}}
\newcommand{\cT}{\mathcal{T}}
\newcommand{\RR}{\mathbb{R}}
\newcommand{\norm}[1]{\left\lVert #1\right\rVert}
\newcommand{\abs}[1]{\left\lvert #1\right\rvert}
\newcommand{\e}{\mathrm{e}}
\newcommand{\Ai}{A_{\infty}}
\newcommand{\Aizero}{A_{\infty,0}}

\newcommand{\PublicRepositoryURL}{https://github.com/kevin-goldstein/einstein_weyl}

\newcommand{\CertifiedBLower}{0.3633018786279300967082111952175988891934060723442491092}
\newcommand{\CertifiedBUpper}{0.3633018786279300967082111952175988891934060823442491093}
\newcommand{\CertifiedRLower}{0.6972243957177816041843153424327653277909997011594910385}
\newcommand{\CertifiedRUpper}{0.6972243957177816041843153424337653277909997011594910386}
\newcommand{\CertifiedMassLower}{0.2755350525036245596846325642399697}
\newcommand{\CertifiedMassUpper}{0.2755350525036245596846325661930532}
\newcommand{\CertifiedTemperatureLower}{0.1178424927965176185317560526720768}
\newcommand{\CertifiedTemperatureUpper}{0.1178424927965176185317560526790940}
\newcommand{\CertifiedEntropyLower}{-5.42428567945554504636589313312624554368829707376709}
\newcommand{\CertifiedEntropyUpper}{-5.42428567945554504636589313312624554368829701093523}

\title{A computer assisted existence proof for a non-Schwarzschild black hole in Einstein--Weyl gravity}

\author{Kevin Goldstein$^{*,\dagger}$, Vishnu Jejjala$^{*,\dagger}$}

\affiliation[*]{Mandelstam Institute for Theoretical Physics, School of Physics, University of the Witwatersrand, Johannesburg 2050, South Africa}
\affiliation[\dagger]{National Institute for Theoretical and Computational Sciences, Gauteng, South Africa}

\emailAdd{kevin.goldstein@wits.ac.za}
\emailAdd{v.jejjala@wits.ac.za}

\abstract{Static, asymptotically flat, non-Schwarzschild black holes in four-dimensional Einstein--Weyl gravity have previously been constructed by numerical shooting, while the local field equations admit horizon and asymptotic expansions to all orders.
We give a computer assisted existence proof for a non-Schwarzschild solution whose dimensionless horizon radius is one in units of the massive spin-two scale, with parameters close to the previously reported numerical branch.
The exterior is divided into a convergent horizon series domain, a compact core enclosed by a parameter dependent Taylor model with directed rounding, and an infinite asymptotic domain treated as a centered stable manifold fixed point problem.
The resulting five-dimensional matching map between the core and tail is continuous on a rigorously specified box.
After exact rational preconditioning, the signs on opposite faces satisfy the hypotheses of the Poincar\'e--Miranda theorem and therefore certify a matching zero.
The solution has a regular nondegenerate horizon, no additional exterior horizon, vanishing Ricci scalar, and nonvanishing Ricci tensor.
The strict bounds also imply existence at every sufficiently nearby horizon radius and a local continuous family at fixed coupling.
The proof establishes existence, but not uniqueness of the matching zero.}


\begin{document}
\maketitle
\flushbottom

\parskip=.33\baselineskip
\renewcommand\baselinestretch{1.15}

\section{Introduction}
\label{sec:introduction}

Quadratic curvature terms are natural both in perturbatively renormalizable models of gravity and in low energy effective actions descended from ultraviolet completions~\cite{Stelle:1977ry,Stelle:1978}.
They change the classical solution space because the metric field equations are fourth order and propagate, in addition to the massless graviton, massive modes whose static fields have Yukawa behavior.
A particularly economical four-dimensional model is Einstein--Weyl gravity, in which the Einstein--Hilbert action is supplemented by the square of the Weyl tensor.
Because $m_2\bar r_h=1$ probes the higher derivative scale, the theorem concerns the stated classical Einstein--Weyl equations and does not justify neglecting further curvature operators in a generic effective action~\cite{Goldstein:2017rxu}.

The static and spherically symmetric vacuum problem in this theory has a branch that is not Ricci-flat.
L\"u, Perkins, Pope, and Stelle demonstrated this branch by numerical shooting and showed that every solution in the relevant asymptotically flat sector has vanishing Ricci scalar~\cite{Lu:2015cqa,Lu:2015psa}.
The branch meets the Schwarzschild family at a static negative mode of the Lichnerowicz operator~\cite{Lu:2017kzi}.
Subsequent work developed asymptotic transseries, accurate continued fraction representations, and exact local series to all orders~\cite{Goldstein:2017rxu,Kokkotas:2017zwt,Podolsky:2018pfe,Podolsky:2019gro}.
In particular, Podolsk\'y, \v{S}varc, Pravda, and Pravdov\'a use a metric conformal to Kundt geometry to reduce the problem to two compact ordinary differential equations and identify a dimensionless horizon Bach parameter $b$~\cite{Podolsky:2019gro}.
The value $b=0$ gives Schwarzschild, whereas $b\neq0$ gives a nonvanishing Bach tensor at the horizon.

The remaining analytic difficulty is global.
A generic regular horizon solution contains a massive mode that grows like $\e^{+m_{2}\bar r}$ and destroys asymptotic flatness.
The asymptotically flat solution is selected by suppressing that mode to all orders.
At the same time, the decaying mode is beyond all algebraic orders in $1/\bar r$, so an ordinary Frobenius expansion at spatial infinity sees only Schwarzschild~\cite{Goldstein:2017rxu,Podolsky:2019gro}.
This combination makes direct outward shooting exponentially sensitive to initial data and prevents the local series, by themselves, from proving that a regular horizon connects to an asymptotically flat end.

The purpose of this paper is to close that global gap by a computer assisted existence proof at a fixed dimensionless horizon radius.
The construction has four components.
First, an exact majorant proves convergence of the horizon recurrence and controls its derivatives.
Second, a centered Taylor model with outward rounding encloses the compact part of the exterior together with its dependence on the two horizon side parameters $(b,R)$.
Third, a generalized Yukawa expansion at infinity is corrected by a Banach fixed point argument in a weighted function space.
Fourth, exact rational bounds certify the signs of a preconditioned five-dimensional matching map on opposite faces of a small box, so the Poincar\'e--Miranda theorem supplies a zero~\cite{Miranda:1940,Kulpa:1997}.
The logical structure of the proof is worth emphasizing.
The horizon majorant gives a rigorous local analytic solution.
The validated Taylor model flow then gives existence and uniqueness of the compact core evolution for each admissible horizon parameter.
Independently, the stable manifold fixed point theorem gives a unique tail in the specified weighted ball for each admissible set of asymptotic parameters.
The final Poincar\'e--Miranda argument connects these two constructions by proving the existence of a parameter tuple for which the core and tail agree.
This separation distinguishes the analytic existence statements from the validated numerical inequalities used to complete the global connection problem.
The use of validated Taylor methods follows the established framework for rigorous initial value integration and for reducing dependency and wrapping~\cite{Nedialkov:1999,Berz:1998,Neher:2007}.
For related variational flow validation and a posteriori contraction methods, see~\cite{Zgliczynski:2002,vandenBerg:2015}.
Computer assisted arguments have also proved global existence statements in nonlinear general relativity in substantially more singular settings~\cite{Reiterer:2019}.

Our quantitative existence theorem certifies a static, spherical, asymptotically flat Einstein--Weyl black hole with
\begin{equation}
 m_{2}\bar r_{h}=1\,.
 \label{eq:rhoh-one-intro}
\end{equation}
The proof encloses its Bach parameter and compactified endpoint in the intervals
\begin{equation}
\begin{split}
 b\in{}&
 \big[
 \CertifiedBLower\,,
 \\[-1mm]
 &\hspace{16mm}
 \CertifiedBUpper
 \big]\,,
 \label{eq:b-box-intro}
\end{split}
\end{equation}
\begin{equation}
\begin{split}
 R\in{}&
 \big[
 \CertifiedRLower\,,
 \\[-1mm]
 &\hspace{16mm}
 \CertifiedRUpper
 \big]\,.
 \label{eq:R-box-intro}
\end{split}
\end{equation}
Here $R$ is the finite value of the horizon expansion coordinate at spatial infinity.
The small radii reflect exponential amplification in the core: derivatives of the endpoint with respect to $b$ are of order $10^{14}$, while the tail tolerance is of order $10^{-30}$.
They are numerical conditioning requirements, not observational precision or a uniqueness criterion.
The proof also establishes positivity of both metric functions throughout the exterior, so there is no additional exterior Killing horizon.
Because the interval~\eqref{eq:b-box-intro} excludes zero, the solution has nonzero Bach tensor and is not Ricci-flat.

All numerical bounds used in the compact core are computed with directed rounding, including error enclosures for the approximate stored centers, and every subsequent matching inequality is checked with exact rational arithmetic.
The resulting metric is represented by a convergent local series, a rigorously enclosed ODE orbit, and a stable asymptotic fixed point joined at a certified parameter value.
The same strict bounds establish a local continuous family of non-Schwarzschild solutions at fixed coupling, rather than only a single solution, as proved in Proposition~\ref{prop:continuous-family}.
Section~\ref{sec:scope} explains the distinction between this constructive representation and a closed form metric, and the additional estimates that would be needed for uniqueness.

The paper is organized as follows.
Section~\ref{sec:equations} reviews the Einstein--Weyl equations and the Schwarzschild--Bach horizon expansion.
Section~\ref{sec:infinity} regularizes spatial infinity and derives the conserved constraint and the stable massive sector.
Section~\ref{sec:proof} gives the computer assisted proof.
Section~\ref{sec:geometry} extracts the global geometry and physical parameters and proves the existence of a local continuous family.
Section~\ref{sec:cold-endpoint} clarifies the scope of possible cold limits and proves a local obstruction to a regular extremal horizon of finite area.
Section~\ref{sec:discussion} summarizes the result, its limitations, and possible extensions to quantitative branch continuation and stability.
The appendices record the recurrence, majorant, core and tail estimates, matching certificate, and reproducibility information.

\section{Einstein--Weyl equations and the horizon expansion}
\label{sec:equations}

\subsection{Action and spherical reduction}

We consider the action
\begin{equation}
 S[g]
 =
 \frac{1}{16\pi G}
 \int \dd^{4}x\,\sqrt{-g}
 \left(
 R-k C_{\mu\nu\rho\sigma}C^{\mu\nu\rho\sigma}
 \right)\,,
 \qquad
 k>0\,.
 \label{eq:action}
\end{equation}
With the Bach tensor convention
\begin{equation}
 B_{\mu\nu}
 =
 \left(
 \nabla^{\rho}\nabla^{\sigma}
 +\frac{1}{2}R^{\rho\sigma}
 \right)
 C_{\mu\rho\nu\sigma}\,,
 \label{eq:bach-definition}
\end{equation}
the vacuum equation is
\begin{equation}
 G_{\mu\nu}-4kB_{\mu\nu}=0\,.
 \label{eq:field-equation}
\end{equation}
The Bach tensor is traceless, and therefore
\begin{equation}
 R=0\,.
 \label{eq:R-zero}
\end{equation}
Equation~\eqref{eq:field-equation} then becomes
\begin{equation}
 R_{\mu\nu}=4kB_{\mu\nu}\,.
 \label{eq:Ricci-Bach}
\end{equation}
We define the massive spin-two scale by
\begin{equation}
 m_{2}^{2}=\frac{1}{2k}\,.
 \label{eq:m2}
\end{equation}

Following the conformal construction of~\cite{Pravda:2017} and its spherical specialization in~\cite{Podolsky:2018pfe,Podolsky:2019gro}, we use the following ansatz conformal to Kundt geometry:
\begin{equation}
 \dd s^{2}
 =
 \Omega(r)^{2}
 \left[
 \dd\Omega_{2}^{2}
 -2\,\dd u\,\dd r
 +\mathcal{H}(r)\,\dd u^{2}
 \right]\,.
 \label{eq:kundt-metric}
\end{equation}
Primes in this subsection denote differentiation with respect to $r$.
The complete spherical vacuum system reduces to
\begin{equation}
 \Omega\Omega''-2\Omega'^{2}
 =
 \frac{k}{3}\mathcal{H}''''\,,
 \label{eq:kundt-one}
\end{equation}
\begin{equation}
 \Omega\Omega'\mathcal{H}'
 +3\Omega'^{2}\mathcal{H}
 +\Omega^{2}
 =
 \frac{k}{3}
 \left(
 \mathcal{H}'\mathcal{H}'''
 -\frac{1}{2}\mathcal{H}''^{2}
 +2
 \right)\,.
 \label{eq:kundt-two}
\end{equation}
The trace equation is
\begin{equation}
 \mathcal{H}\Omega''
 +\mathcal{H}'\Omega'
 +\frac{1}{6}
 \left(
 \mathcal{H}''+2
 \right)
 \Omega
 =0\,.
 \label{eq:kundt-trace}
\end{equation}
Only two of~\eqref{eq:kundt-one}--\eqref{eq:kundt-trace} are independent, but retaining the remaining equation as a constraint will be useful.

\subsection{Dimensionless horizon variables}

Let $z=0$ denote a nondegenerate horizon and normalize the dimensionless areal radius to
\begin{equation}
 \rho=m_{2}\bar r=W(z)\,,
 \qquad
 W(0)=\rho_{h}=1\,.
 \label{eq:rho-W}
\end{equation}
Fix the Kundt coordinate freedom by $r_h=-1/\bar r_h=-m_2$ and set
\begin{equation}
 z=1-\frac{r}{r_h}\,,\qquad r=m_2(z-1)\,,\qquad
 \Omega=\frac{W}{m_2}\,,\qquad \mathcal H=-m_2^2zQ\,.
 \label{eq:coordinate-dictionary}
\end{equation}
With $\upsilon=m_2u$, the regular metric is
\begin{equation}
 m_2^2\dd s^2=W^2\bigl[\dd\Omega_2^2-2\dd\upsilon\,\dd z-zQ\dd\upsilon^2\bigr]\,.
 \label{eq:dimensionless-kundt}
\end{equation}
Thus the dimensionless reduced equations use $\Omega=W$, $\mathcal H=-zQ$, and coupling $\kappa=km_2^2=1/2$.
These identifications refer to the rescaled metric in~\eqref{eq:dimensionless-kundt}; the physical quantities are given by~\eqref{eq:coordinate-dictionary}.
The exact horizon family can be written as
\begin{equation}
 W(z)
 =
 \frac{1}{1-z}
 +b\sum_{n=1}^{\infty}\alpha_{n}z^{n}\,,
 \label{eq:W-horizon}
\end{equation}
\begin{equation}
 Q(z)
 =
 (1-z)^{2}
 +3b\sum_{n=1}^{\infty}(-1)^{n}\gamma_{n}z^{n}\,.
 \label{eq:Q-horizon}
\end{equation}
The coefficient $b$ is the dimensionless Bach parameter at the horizon.
For $b=0$,~\eqref{eq:W-horizon} and~\eqref{eq:Q-horizon} give the Schwarzschild solution.
For $b\neq0$, the Bach invariant at the horizon is nonzero~\cite{Podolsky:2019gro}.

In the horizon time normalization, $\tau=u-\int\mathcal H^{-1}\dd r$ has units of length and the areal coordinate metric is
\begin{equation}
 \dd s^{2}
 =
 -\widehat h(\rho)\,\dd\tau^{2}
 +\frac{\dd\rho^{2}}{m_2^2 f(\rho)}
 +\frac{\rho^{2}}{m_{2}^{2}}\dd\Omega_{2}^{2}\,,
 \label{eq:areal-metric}
\end{equation}
where
\begin{equation}
 \widehat h=zQW^{2}\,,
 \qquad
 f=zQ\left(\frac{W_{z}}{W}\right)^{2}\,.
 \label{eq:metric-functions-horizon}
\end{equation}
The coefficients in~\eqref{eq:W-horizon} and~\eqref{eq:Q-horizon} satisfy the exact triangular recurrence recorded in Appendix~\ref{app:recurrence}.
The first coefficients are
\begin{equation}
 \alpha_{0}=0\,,
 \qquad
 \alpha_{1}=1\,,
 \qquad
 \gamma_{1}=1\,,
 \qquad
 \gamma_{2}=1+b\,,
 \label{eq:first-coefficients}
\end{equation}
where we used $\rho_{h}=1$.
In particular,
\begin{equation}
 W_{z}(0)=1+b\,,
 \qquad
 Q(0)=1\,.
 \label{eq:horizon-slopes}
\end{equation}
The intervals in Theorem~\ref{thm:main} imply $W_{z}(0)>0$, so both metric functions have a simple zero at the horizon.

\subsection{Statement of the main theorem}

\begin{theorem}[Computer assisted existence]
\label{thm:main}
Consider Einstein--Weyl gravity defined by~\eqref{eq:action}, with $k>0$ and $m_{2}^{2}=1/(2k)$.
There exists a static and spherically symmetric vacuum solution with $m_{2}\bar r_{h}=1$ whose horizon Bach parameter lies in the interval~\eqref{eq:b-box-intro} and whose compactified endpoint lies in the interval~\eqref{eq:R-box-intro}.
The certified exterior admits an analytic ingoing extension across a nondegenerate Killing horizon, which is the future event horizon relative to the chosen asymptotically flat end in the spacetime obtained by adjoining a sufficiently small interior collar.
The exterior is smooth and asymptotically flat after a constant normalization of the time coordinate.
The exterior metric functions obey
\begin{equation}
 h(\rho)>0\,,
 \qquad
 f(\rho)>0\,,
 \qquad
 \rho>1\,,
 \label{eq:metric-positivity-theorem}
\end{equation}
so there is no additional exterior Killing horizon.
Moreover,
\begin{equation}
 R=0\,,
 \qquad
 R_{\mu\nu}\neq0\,.
 \label{eq:curvature-theorem}
\end{equation}
The theorem asserts existence of at least one solution in the parameter box and does not assert uniqueness of the matching zero.
\end{theorem}

The proof occupies Section~\ref{sec:proof}.
Before giving it, we transform the irregular endpoint at spatial infinity into a form suitable for a stable manifold argument.

\section{Regularizing spatial infinity}
\label{sec:infinity}

\subsection{Removing the pole and the double zero}

Let the horizon coordinate reach spatial infinity at the finite value $z=R$.
Define
\begin{equation}
 x=\frac{z}{R}\,,
 \qquad
 s=1-x\,,
 \qquad
 W=\frac{1}{sA(x)}\,,
 \qquad
 Q=Rs^{2}B(x)\,.
 \label{eq:regular-variables}
\end{equation}
The required simple pole of $W$ and double zero of $Q$ are explicit in~\eqref{eq:regular-variables}.
Spatial infinity is $x=1$.
The first field equation and the trace equation become
\begin{equation}
\begin{split}
 sA''-2A'
 ={}&
 \frac{\kappa}{3}s^{3}A^{3}
 \left[
 xs^{2}B''''
 +4s(1-3x)B'''
 \right.
 \\
 &\left.
 \hspace{22mm}
 +12(3x-2)B''
 +24B'
 \right]\,,
 \label{eq:x-first}
\end{split}
\end{equation}
\begin{equation}
\begin{split}
 0={}&
 A^{2}
 \left[
 xs^{2}B''
 +2sB'
 +2(B-1)
 \right]
 \\
 &-6s^{2}AB(xA''+A')
 -6xs^{2}AA'B'
 +12xs^{2}B(A')^{2}\,,
 \label{eq:x-trace}
\end{split}
\end{equation}
where primes in~\eqref{eq:x-first} and~\eqref{eq:x-trace} denote $x$ derivatives and
\begin{equation}
 \kappa=\frac{1}{2\rho_{h}^{2}}=\frac{1}{2}\,.
 \label{eq:kappa}
\end{equation}

The final endpoint regularization is
\begin{equation}
 t=\frac{1}{1-x}\,,
 \qquad
 B(x)=\frac{t+C(t)}{t-1}\,.
 \label{eq:t-C}
\end{equation}
The horizon is $t=1$, and spatial infinity is $t=+\infty$.
Equations~\eqref{eq:x-first} and~\eqref{eq:x-trace} reduce exactly to
\begin{equation}
 tA''
 =
 \frac{\kappa}{3}A^{3}C''''\,,
 \label{eq:t-first}
\end{equation}
\begin{equation}
\begin{split}
 0={}&
 -6t^{2}AA''
 +12t^{2}(A')^{2}
 +tA^{2}C''
 -6tACA''
 \\
 &-6tAA'C'
 -12tAA'
 +12tC(A')^{2}
 +2A^{2}C'
 -6ACA'\,.
 \label{eq:t-trace}
\end{split}
\end{equation}
Primes henceforth denote $t$ derivatives unless stated otherwise.

Set
\begin{equation}
 p=A'\,,
 \qquad
 d=C'\,,
 \qquad
 Y=C''\,,
 \qquad
 J=C'''\,.
 \label{eq:first-order-variables}
\end{equation}
Solving~\eqref{eq:t-trace} for $A''$ gives
\begin{equation}
\begin{split}
 F(A,p,C,d,Y;t)
 :={}&A''
 \\
 ={}&
 \frac{2p^{2}}{A}
 +\frac{AY}{6(t+C)}
 +\frac{Ad}{3t(t+C)}
 -\frac{pd}{t+C}
 \\
 &-\frac{Cp}{t(t+C)}
 -\frac{2p}{t+C}\,.
 \label{eq:F}
\end{split}
\end{equation}
Equation~\eqref{eq:t-first} then gives
\begin{equation}
 J'
 =
 \frac{3t}{\kappa A^{3}}F\,.
 \label{eq:J-prime}
\end{equation}
Together with
\begin{equation}
 A'=p\,,
 \qquad
 C'=d\,,
 \qquad
 d'=Y\,,
 \qquad
 Y'=J\,,
 \label{eq:first-order-rest}
\end{equation}
\eqref{eq:F}--\eqref{eq:first-order-rest} form the six-dimensional core system used in the validated integration.

\subsection{The conserved field equation constraint}

Using only~\eqref{eq:t-first} and~\eqref{eq:t-trace} would be insufficient unless the remaining independent Einstein--Weyl equation were controlled.
The needed control is algebraic.

\begin{lemma}[Constraint propagation]
\label{lem:constraint}
Define
\begin{equation}
 \cE_{2}=-\frac{\cN}{6A^{4}}\,,
 \label{eq:E2-constraint}
\end{equation}
where
\begin{equation}
\begin{split}
 \cN={}&
 2\kappa t^{2}A^{4}dJ
 -\kappa t^{2}A^{4}Y^{2}
 -4\kappa t^{2}A^{4}J
 -6\kappa tA^{4}CJ
 \\
 &+2\kappa tA^{4}dY
 +8\kappa tA^{4}Y
 +6\kappa A^{4}CY
 -4\kappa A^{4}d^{2}
 -8\kappa A^{4}d
 \\
 &+18t^{4}p^{2}
 -6t^{3}Apd
 -24t^{3}Ap
 +18t^{3}Cp^{2}
 +6t^{2}A^{2}d
 -18t^{2}ACp\,.
 \label{eq:N-constraint}
\end{split}
\end{equation}
Along every solution of~\eqref{eq:F},~\eqref{eq:J-prime}, and~\eqref{eq:first-order-rest},
\begin{equation}
 \cN'=4\frac{p}{A}\cN\,,
 \label{eq:N-propagation}
\end{equation}
and hence
\begin{equation}
 \cE_{2}'=0\,.
 \label{eq:E2-conserved}
\end{equation}
The condition $\cE_{2}=0$ is~\eqref{eq:kundt-two}, expressed in the dimensionless variables after use of the first equation and the trace equation.
Its normalization is fixed by~\eqref{eq:E2-constraint}.
\end{lemma}

\begin{proof}
Differentiate~\eqref{eq:N-constraint} and substitute the six first order equations.
Exact symbolic simplification gives~\eqref{eq:N-propagation}.
The derivative of $A^{-4}\cN$ therefore vanishes.
The ancillary symbolic check also substitutes $W=t/A$ and the coordinate transformation into the residual of~\eqref{eq:kundt-two} and verifies its equality to $-\cN/(6A^4)$.
\end{proof}

\begin{corollary}
\label{cor:constraint}
Suppose a solution of the first order system has $A\to\Ai>0$, $C\to\beta\in\RR$, and
\begin{equation}
 |p|+|d|+|Y|+|J|=O(\e^{-\sigma t})\,,\qquad \sigma>0\,,
 \label{eq:constraint-decay}
\end{equation}
and its denominators remain nonzero.
Then it satisfies the full reduced Einstein--Weyl equations.
\end{corollary}

\begin{proof}
Every term in~\eqref{eq:N-constraint} is a bounded coefficient times a polynomial in $t$ and at least one exponentially decaying derivative variable.
Hence $\cN\to0$ and $\cE_2\to0$.
Its conservation gives $\cE_2=0$, and therefore $\cN=0$, at every finite $t$.
Merely requiring the derivative variables to tend to zero would not suffice because of the powers of $t$ in $\cN$.
\end{proof}

\subsection{Stable and unstable massive modes}

Let
\begin{equation}
 A(t)\longrightarrow\Ai\,,
 \qquad
 C(t)\longrightarrow\beta\,,
 \qquad
 t\longrightarrow\infty\,,
 \label{eq:asymptotic-constants}
\end{equation}
and define
\begin{equation}
 a=\frac{1}{\Ai}\,.
 \label{eq:a-definition}
\end{equation}
The dimensionless areal radius is
\begin{equation}
 \rho=W=\frac{t}{A(t)}\sim at\,.
 \label{eq:rho-asymptotic}
\end{equation}
The linearized massive sector is diagonalized by
\begin{equation}
 u=J+aY\,,
 \qquad
 v=J-aY\,.
 \label{eq:uv-variable}
\end{equation}
Writing
\begin{equation}
 H
 =
 \frac{3t}{\kappa A^{3}}F-a^{2}Y\,,
 \label{eq:H}
\end{equation}
one obtains
\begin{equation}
 u'=au+H\,,
 \qquad
 v'=-av+H\,.
 \label{eq:uv-system}
\end{equation}
The condition that removes the growing mode is $u(\infty)=0$.
The decaying homogeneous mode is $v\sim\e^{-at}$.

The generalized Yukawa expansion has a power correction.
Consider the ansatz
\begin{equation}
 \e^{-at}t^{\nu}
 \sum_{n=0}^{\infty}\frac{X_{n}}{t^{n}}\,.
 \label{eq:generalized-Yukawa}
\end{equation}
Substitution into the linearized system gives the solvability condition
\begin{equation}
 \nu=\frac{\beta}{2\Ai}\,.
 \label{eq:nu}
\end{equation}
As shown below, the dimensionless mass parameter is (with $M=G M_{\rm ADM}$ the geometrized ADM mass)
\begin{equation}
 \mu=m_{2}M=-\frac{\beta}{2\Ai}\,,
 \label{eq:mu-definition}
\end{equation}
so the stable mode has the expected power $t^{-\mu}\e^{-at}$.

For the uniform tail certificate we use the exact rational centers supplied in \path{tail_shape_data.py}, with abbreviated values
\begin{equation}
\begin{split}
 \Aizero&\simeq
 1.0324864466507795973677035145514604182115606537271511328\,,
 \\
 \beta_{0}&\simeq
 -0.5689724145744066285521979945096614011267654811607729603\,,
 \label{eq:tail-centers}
\end{split}
\end{equation}
and choose the exact decimal rational approximants from the ancillary data,
\begin{equation}
 |a_0-\Aizero^{-1}|<10^{-178}\,,\qquad
 \left|\nu_0-\frac{\beta_0}{2\Aizero}\right|<10^{-178}\,.
 \label{eq:a0-nu0}
\end{equation}
The residual calculation includes these discrepancies.
The fixed matching coordinate is
\begin{equation}
 v_{0}=J-a_{0}Y\,.
 \label{eq:v0}
\end{equation}
Variations of $\Ai$ and $\beta$ away from their central values are included in the rigorous defect bounds rather than in the diagonalizing coordinate.

\subsection{Metric functions at infinity}

Equations~\eqref{eq:regular-variables} and~\eqref{eq:t-C} give
\begin{equation}
 \widehat h
 =
 \frac{R^{2}}{A^{2}}
 \left(1+\frac{C}{t}\right)\,,
 \label{eq:hhat-tail}
\end{equation}
\begin{equation}
 f
 =
 \left(1+\frac{C}{t}\right)
 \left(1-\frac{tp}{A}\right)^{2}\,.
 \label{eq:f-tail}
\end{equation}
The horizon normalized temporal function approaches $R^{2}/\Ai^{2}$.
After normalizing time at infinity,
\begin{equation}
 h
 =
 \frac{\Ai^{2}}{A^{2}}
 \left(1+\frac{C}{t}\right)\,.
 \label{eq:h-tail}
\end{equation}
Since $\rho=t/A$, the stable tail estimates imply, for every fixed $0<\sigma'<0.6$,
\begin{equation}
 h
 =
 1-\frac{2\mu}{\rho}
 +O\!\left(\e^{-\sigma' t}\right)\,,
 \qquad
 f
 =
 1-\frac{2\mu}{\rho}
 +O\!\left(\e^{-\sigma' t}\right)\,,
 \label{eq:asymptotic-flatness}
\end{equation}
where $\mu$ is given by~\eqref{eq:mu-definition}.
The slightly smaller exponent absorbs the factor $t$ multiplying $p$ in~\eqref{eq:f-tail}; polynomial factors times $\e^{-0.6t}$ are $O(\e^{-\sigma't})$.
This is the asymptotically flat end required by Theorem~\ref{thm:main}.

\section{Computer assisted proof}
\label{sec:proof}

\subsection{Domain decomposition}

The proof divides the exterior into the three domains, joined at their common endpoints, shown in Figure~\ref{fig:domains}.
The horizon series covers
\begin{equation}
 0\leq z\leq\frac{R}{20}\,,
 \qquad
 1\leq t\leq\frac{20}{19}\,.
 \label{eq:horizon-domain}
\end{equation}
The validated compact core covers
\begin{equation}
 \frac{20}{19}\leq t\leq40\,.
 \label{eq:core-domain}
\end{equation}
The centered stable tail covers
\begin{equation}
 40\leq t<\infty\,.
 \label{eq:tail-domain}
\end{equation}

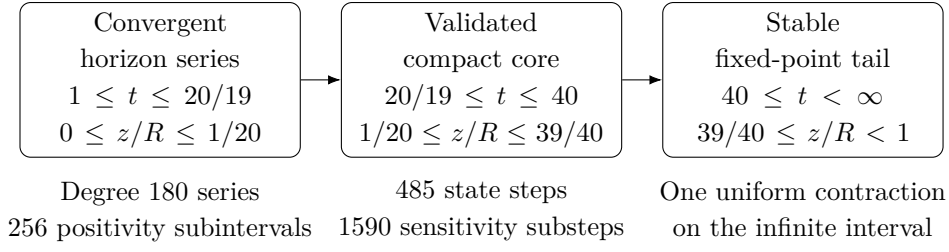
\begin{figure}[t]
\centering
\begin{tikzpicture}[>=Latex,
 proofbox/.style={draw,rounded corners,align=center,text width=3.45cm,minimum height=1.5cm,font=\small}]
 \node[proofbox] (h) at (0,0) {Convergent horizon series\\$1\leq t\leq20/19$\\$0\leq z/R\leq1/20$};
 \node[proofbox] (c) at (4.25,0) {Validated\\compact core\\$20/19\leq t\leq40$\\$1/20\leq z/R\leq39/40$};
 \node[proofbox] (a) at (8.5,0) {Stable\\fixed-point tail\\$40\leq t<\infty$\\$39/40\leq z/R<1$};
 \draw[->] (h) -- (c); \draw[->] (c) -- (a);
 \node[align=center,font=\small,below=5pt] at (h.south) {Degree 180 series\\256 positivity subintervals};
 \node[align=center,font=\small,below=5pt] at (c.south) {485 state steps\\1590 sensitivity substeps};
 \node[align=center,font=\small,below=5pt] at (a.south) {One uniform contraction\\on the infinite interval};
\end{tikzpicture}
\caption{The validated domain decomposition (the horizontal spacing is schematic).
The series--core join is $z=R/20$, and the core--tail match is $z=39R/40$.}
\label{fig:domains}
\end{figure}

The five matching parameters are
\begin{equation}
 \theta=(b,R,\Ai,\beta,\eta)\,,
 \label{eq:theta}
\end{equation}
where $\eta$ is the stable amplitude normalized by $v_{0}(40)=\eta$.
Their centers and radii are listed in Table~\ref{tab:parameter-box}.

\begin{table}[t]
\centering
\begin{tabular}{@{}lll@{}}
\toprule
Parameter & Center & Radius \\
\midrule
$b$
& $0.3633018786279300967082111952175988891934060773442491092$
& $5\times10^{-45}$ \\
$R$
& $0.6972243957177816041843153424332653277909997011594910385$
& $5\times10^{-31}$ \\
$\Ai$
& $1.0324864466507795973677035145514604182$
& $3\times10^{-29}$ \\
$\beta$
& $-0.5689724145744066285521979945096614011$
& $2\times10^{-27}$ \\
$\eta$
& $1.8427814530748829568597814537921\times10^{-17}$
& $2\times10^{-29}$ \\
\bottomrule
\end{tabular}
\caption{The five-dimensional parameter box used in the matching theorem.
The displayed centers are abbreviated approximations, not the definitions of the box.
The machine readable certificate uses the full exact decimal strings supplied in the ancillary files.
Equations~\eqref{eq:b-box-intro}--\eqref{eq:R-box-intro} give outward decimal enclosures of the resulting $b$ and $R$ intervals.}
\label{tab:parameter-box}
\end{table}

\subsection{Convergence and positivity of the horizon expansion}

The exact recurrence admits a uniform complex majorant.

\begin{proposition}[Horizon majorant]
\label{prop:horizon-majorant}
For complex $b$ with $\abs{b}\leq0.37$, write
\begin{equation}
 W(z)=\sum_{n=0}^{\infty}w_{n}z^{n}\,,
 \qquad
 Q(z)=\sum_{n=0}^{\infty}q_{n}z^{n}\,.
 \label{eq:wq-coefficients}
\end{equation}
Then
\begin{equation}
 \abs{w_{n}}
 \leq
 \frac{10^{n}}{(n+1)^{2}}\,,
 \qquad
 \abs{q_{n}}
 \leq
 \frac{5\,10^{n}}{(n+1)^{4}}\,.
 \label{eq:majorant}
\end{equation}
Consequently, both series converge absolutely for $\abs{z}<0.1$, and their first four derivatives have explicit geometric tail bounds on every smaller disk.
\end{proposition}

\begin{proof}
The direct coefficient recurrences and the well founded order $w_2,q_3,w_3,q_4,\ldots$ are given in Appendix~\ref{app:recurrence}.
The initial coefficients satisfy the stated bounds.
For each subsequent step, the appendix derives the rational ratio functions and proves their high order estimates by explicit split sums.
The remaining finite range is checked exactly by the ancillary certificate.
The resulting coefficient bounds give normal convergence on every compact subdisk of $|z|<0.1$.
\end{proof}

The complete parameter interval~\eqref{eq:R-box-intro} satisfies
\begin{equation}
 \frac{R}{20}<0.035\,,
 \label{eq:zjoin-small}
\end{equation}
so the series to core join lies well inside the certified disk.
Using $180$ coefficients, outward interval evaluation on $256$ subintervals proves the following bounds throughout the horizon series domain:
\begin{equation}
 W\geq1\,,
 \label{eq:W-horizon-positive}
\end{equation}
\begin{equation}
 W_{z}
 \geq
 1.3633018786279300967082111952175988\,,
 \label{eq:Wz-horizon-positive}
\end{equation}
\begin{equation}
 Q
 \geq
 0.89528942067630256330400902383933075\,.
 \label{eq:Q-horizon-positive}
\end{equation}
The analytic remainder bounds are below $7\times10^{-88}$ for $W$ and below $2\times10^{-91}$ for $Q$ there.
The bound $W\ge1$ follows from $W(0)=1$ and the strictly positive derivative bound.
These inequalities exclude a further zero before the core begins.
Analytic horizon extension follows from the regular Kundt form, as explained in Section~\ref{sec:geometry}.

\subsection{Validated parameter dependent core}

The convergent horizon series initializes the first order system at $t=20/19$ for the complete $(b,R)$ box.
The centered parameter scaling is specified in Appendix~\ref{app:core}.
An MPFR interval Taylor method of order 56 at 512 bits encloses the center orbit, while binary128 interval Taylor series of order 12 retain the two first variation columns.
Uniform bounds on the three second variations control the remaining parameter dependence.
Each coarse step has a validated Picard tube.
All operations contributing to an enclosure are rounded outward; point estimates serve only to select stored centers or positive weights.

Appendix~\ref{app:core} gives the initialization bounds, Taylor remainders, logarithmic norm estimates, weight changes, and parameter bootstrap inequalities.
The sensitivity tube is computed from the interval coefficients of the true center Jacobian, and all recentering radii are measured from the represented center to both interval endpoints.
The maximal verified bootstrap ratios are bounded above by
\begin{equation}
 0.004006962574346467\,,\qquad 0.002383324686807518\,,
 \label{eq:core-bootstrap}
\end{equation}
so the uniform state and first variation assumptions are strict self enclosures.

The geometric bounds throughout the core are
\begin{align}
 A&\ge1.00235962122891209481142472124\,,
 \label{eq:A-core-positive}\\
 t+C&\ge0.0737302366694336803815908068962\,,
 \label{eq:D-core-positive}\\
 U:=1-\frac{tp}{A}
 &\ge0.950881733891050506912121315795\,.
 \label{eq:U-core-positive}
\end{align}
They separate the denominators of the differential equations from zero and imply positivity of both metric functions.
Table~\ref{tab:core-endpoint} summarizes the center state endpoint.
The complete parameter dependent enclosure also includes the two interval first variation columns and the quadratic remainder in~\eqref{eq:core-parameter-remainder}.

\begin{table}[t]
\centering
\begin{tabular}{@{}lcc@{}}
\toprule
Component & Approximate midpoint & Width bound\\
\midrule
$A$ & $1.0324864466507795973262935881174168583$
 & $5.89\times10^{-36}$\\
$p$ & $4.1360079890371705\times10^{-20}$
 & $5.90\times10^{-36}$\\
$C$ & $-0.56897241457440663853181033974624473$
 & $1.15\times10^{-33}$\\
$d$ & $9.7268279902030220\times10^{-18}$
 & $1.19\times10^{-33}$\\
$Y$ & $-9.4818074159354655\times10^{-18}$
 & $1.22\times10^{-33}$\\
$J$ & $9.2443453949621782\times10^{-18}$
 & $1.26\times10^{-33}$\\
\bottomrule
\end{tabular}
\caption{Center state enclosure at $t=40$ for $(b,R)=(b_0,R_0)$.
The rounded midpoints are orientation values, not definitions of the enclosures.
The width bounds are rounded upward; the exact outward endpoints are supplied in the ancillary data.
This table excludes parameter displacements.
Their first variation intervals and quadratic remainders enter the matching certificate separately.}
\label{tab:core-endpoint}
\end{table}

\subsection{The centered stable tail}
\label{subsec:tail-proof}

We construct the tail on $t\geq T=40$ in the coordinates
\begin{equation}
 X=(\alpha,p,c,d,u_0,v_0)\,,\qquad
 \alpha=A-\Ai\,,\quad c=C-\beta\,,\quad
 u_0=J+a_0Y\,,\quad v_0=J-a_0Y\,.
 \label{eq:tail-banach-coordinates}
\end{equation}
Here $a_0$ and $\nu_0$ are the fixed rational numbers specified by the finite decimal strings in \path{tail_shape_data.py}.
They approximate $\Aizero^{-1}$ and $\beta_0/(2\Aizero)$, respectively, with errors smaller than $10^{-178}$; these identities are not assumed exact by the certificate.
The physical linearized exponents remain $a=\Ai^{-1}$ and $\nu=\beta/(2\Ai)$.
Using fixed rational approximation exponents makes the entire residual calculation rational, including their discrepancy from the physical exponents.

Write $g(t)=\exp[-a_0(t-T)](t/T)^{\nu_0}$.
The two stored polynomials of degree 40 specify
\begin{equation}
 \alpha_0(t)=\eta g(t)P_A(1/t)\,,\qquad
 c_0(t)=\eta g(t)P_C(1/t)\,.
 \label{eq:tail-center}
\end{equation}
The other approximation components are obtained by exact differentiation: $p_0=\alpha_0'$, $d_0=c_0'$, $Y_0=c_0''$, $J_0=c_0'''$, and the last two components of $X_0$ are $J_0\pm a_0Y_0$.
The polynomials are exactly renormalized so that $J_0(T)-a_0Y_0(T)=\eta$.
For orientation, write $L=(\alpha_0,p_0,c_0,d_0,Y_0,J_0)(T)/\eta$ for the physical shape at the matching point.
Its abbreviated values are
\begin{equation}
 \begin{aligned}
 L\simeq(&-0.0022471425661976992,\ 0.0022444376038924905,
 -0.54155159466056986\,,\\
 &0.52783405074825983,\ -0.51453781456905451,\ 0.50165174929110921)\,.
 \end{aligned}
 \label{eq:stable-shape}
\end{equation}
The exact matching calculation uses the rational polynomial evaluations, not these abbreviated decimal values.
The floating point recurrence is only a way to choose an approximation; the proof independently bounds its residual.

For corrections in the coordinate order~\eqref{eq:tail-banach-coordinates}, use
\begin{equation}
 \norm{Z}_{\sigma,\boldsymbol q}
 =\max_i\sup_{t\geq T}\frac{\e^{\sigma(t-T)}|Z_i(t)|}{q_i}\,,
 \qquad \sigma=0.6\,,
 \label{eq:tail-norm}
\end{equation}
\begin{equation}
 \boldsymbol q=(0.01780,0.01068,3.045,1.827,0.2350,1)\,.
 \label{eq:tail-weights}
\end{equation}
In particular, the last two weights belong to $u_0,v_0$, not $Y,J$.
Appendix~\ref{app:tail} defines the centered operator and its uniform bounds.

\begin{proposition}[Uniform centered tail]
\label{prop:tail}
For the exact central parameters supplied in the ancillary data and
\begin{equation}
 |\Ai-\Aizero|\leq3\times10^{-29}\,,\qquad
 |\beta-\beta_0|\leq2\times10^{-27}\,,\qquad
 |\eta-\eta_0|\leq2\times10^{-29}\,,
 \label{eq:tail-parameter-box}
\end{equation}
the centered operator has a unique fixed point in the ball
\begin{equation}
 \norm{X-X_0}_{\sigma,\boldsymbol q}\leq5\times10^{-30}\,.
 \label{eq:tail-ball}
\end{equation}
This fixed point depends continuously on $(\Ai,\beta,\eta)$, solves all reduced Einstein--Weyl equations, has $v_0(T)=\eta$, and obeys $A>0$, $t+C>0$, and $1-tA'/A>0$ throughout the tail.
Uniqueness is asserted in this ball.
\end{proposition}

\begin{proof}
With the residuals defined in Appendix~\ref{app:tail}, the rational certificate gives
\begin{align}
 |R_F(t)|&\leq
 3.246939858453728996228379\times10^{-33}\e^{-a_0(t-T)}\,,
 \label{eq:RF-bound}\\
 |R_J(t)|&\leq
 6.682954077213574474467859\times10^{-31}\e^{-a_0(t-T)}\,.
 \label{eq:RJ-bound}
\end{align}
On the stated ball the same calculation proves
\begin{equation}
 \norm{D\mathcal T}\leq0.582081478039319\,,
 \label{eq:tail-contraction}
\end{equation}
\begin{equation}
 \norm{\mathcal T(0)}\leq1.468097767726082\times10^{-30}\,.
 \label{eq:tail-defect}
\end{equation}
Thus, with $r=5\times10^{-30}$,
\begin{equation}
 r^{-1}\norm{\mathcal T(0)}+\norm{D\mathcal T}
 \leq0.875701031584536<1\,.
 \label{eq:tail-self-map}
\end{equation}
Banach's theorem applies to the closed ball in the weighted space of continuous functions.
Differentiating the integral equations gives the six ODEs.
All derivative variables decay exponentially, while $A\to\Ai>0$ and $C\to\beta$; hence every term of $\mathcal N$ tends to zero despite its polynomial factors in $t$.
Conservation of $\mathcal N/A^4$ supplies the remaining field equation.
Uniform denominator bounds give continuity of the operator in its parameters in the weighted norm; its uniform contraction then gives continuity of the fixed point.
Finally, the rational lower bounds, rounded downward here, are
\begin{align}
 A&>1.032486446650779597319946\,,
 \label{eq:A-tail-positive}\\
 t+C&>39.43102758542559336\,,
 \label{eq:D-tail-positive}\\
 1-\frac{tA'}A&>0.99999999999999999827\,.
 \label{eq:U-tail-positive}
\end{align}
These imply the geometric assertions and the asymptotically flat end described in~\eqref{eq:asymptotic-flatness}.
\end{proof}

\subsection{The matching map}

Define the validated core endpoint by
\begin{equation}
 \mathcal{X}_{\mathrm c}(b,R)
 =
 (A,p,C,d,Y,J)_{t=40}\,.
 \label{eq:core-map}
\end{equation}
Similarly, denote the exact tail fixed point at the matching surface by
\begin{equation}
 \mathcal{X}_{\mathrm t}(\Ai,\beta,\eta)
 =
 (A,p,C,d,Y,J)_{t=40}\,.
 \label{eq:tail-map}
\end{equation}
We match the five coordinates
\begin{equation}
 \Phi(\theta)
 =
 \begin{pmatrix}
 A_{\mathrm c}-A_{\mathrm t}\\
 p_{\mathrm c}-p_{\mathrm t}\\
 C_{\mathrm c}-C_{\mathrm t}\\
 d_{\mathrm c}-d_{\mathrm t}\\
 v_{0,\mathrm c}-v_{0,\mathrm t}
 \end{pmatrix}\,,
 \qquad
 v_{0}=J-a_{0}Y\,.
 \label{eq:matching-map}
\end{equation}
The horizon majorant, validated ODE flow, and uniform contraction imply that $\Phi$ is continuous on the box in Table~\ref{tab:parameter-box}.

Let $P$ be the exact rational preconditioner obtained by inverting the rationalized midpoint Jacobian, and set
\begin{equation}
 G=P\Phi\,.
 \label{eq:preconditioned-map}
\end{equation}
The Poincar\'e--Miranda theorem applies if, for each coordinate $i$, $G_{i}<0$ on the lower face in direction $i$ and $G_{i}>0$ on the upper face in direction $i$.
All interval information from the core and tail is converted to rational endpoints, allowing one last place decimal unit for every transferred quantity, before this final calculation.
The error budget includes the core endpoint width, derivative column widths times the parameter radii, the quadratic core remainder, and the tail ball component radii; multiplication by $|P|$ gives the five preconditioned errors.
Appendix~\ref{app:poincare} gives the affine decomposition and its exact face test.
Table~\ref{tab:poincare} lists the remaining sign margins as fractions of the corresponding parameter radii.

\begin{table}[t]
\centering
\begin{tabular}{@{}clc@{}}
\toprule
$i$ & Matched parameter direction & Certified margin fraction \\
\midrule
1 & $b$ & $0.540619770207$ \\
2 & $R$ & $0.574495623641$ \\
3 & $\Ai$ & $0.597534509130$ \\
4 & $\beta$ & $0.547484807780$ \\
5 & $\eta$ & $0.647565186013$ \\
\bottomrule
\end{tabular}
\caption{Exact rational Poincar\'e--Miranda face margins.
Entries are rounded downward.
A positive entry means that the combined nominal center displacement and enclosure error occupy less than the face radius by the displayed fraction.}
\label{tab:poincare}
\end{table}

\begin{proposition}[Existence of a matching zero]
\label{prop:matching}
There exists at least one parameter tuple in the box of Table~\ref{tab:parameter-box} for which $\Phi=0$.
\end{proposition}

\begin{proof}
For every $i=1,\ldots,5$, the exact rational certificate gives $G_{i}<0$ on the lower face in direction $i$ and $G_{i}>0$ on the upper face in direction $i$.
The smallest normalized margin is greater than $0.54$.
The Poincar\'e--Miranda theorem therefore gives a point in the box for which $G=0$.
Since $P$ is invertible, $\Phi=0$ at the same point.
\end{proof}

The proposition gives existence rather than uniqueness.
This is why the theorem is stated without a local uniqueness claim.

\subsection{Completion of the sixth state component}

The five coordinate matching map does not explicitly match $Y$ and $J$ separately.
The conserved constraint supplies the missing information.
At fixed $v_{0}=J-a_{0}Y$, write
\begin{equation}
 J=v_{0}+a_{0}Y\,.
 \label{eq:J-from-v}
\end{equation}
Substitute this relation into~\eqref{eq:N-constraint}.
On a common interval box larger than the actual matching neighborhood, exact rational arithmetic gives
\begin{equation}
 -3268.6765117856241
 <
 \left.\frac{\partial\cN}{\partial Y}\right|_{v_{0}}
 <
 -3265.8225581947184\,.
 \label{eq:constraint-monotonicity}
\end{equation}
Thus the constraint is strictly decreasing in $Y$ at fixed $(A,p,C,d,v_{0})$.
The horizon recurrence obeys~\eqref{eq:kundt-two}, as proved in Appendix~\ref{app:recurrence} by evaluating its conserved residual at $z=0$.
This supplies $\cN=0$ at the series--core join, and Lemma~\ref{lem:constraint} propagates it along the core.
The first order $t$ system itself is singular at $t=1$, where $t+C=0$; it is only started at $t=20/19$.
For the tail, the exponential bounds and corollary~\ref{cor:constraint} give $\cN=0$.
Thus both sides of the match obey the constraint.
Matching $(A,p,C,d,v_{0})$ therefore forces their $Y$ values to agree, after which~\eqref{eq:J-from-v} forces $J$ to agree.
The matching zero is a genuine six state solution.

\subsection{Proof of Theorem~\ref{thm:main}}

\begin{proof}[Proof of Theorem~\ref{thm:main}]
Proposition~\ref{prop:horizon-majorant} gives an analytic horizon solution and rigorous initial data for the core.
The horizon positivity bounds~\eqref{eq:W-horizon-positive}--\eqref{eq:Q-horizon-positive} give a regular nondegenerate horizon and no additional zero before $t=20/19$.
The Taylor model integration gives a unique core orbit for every $(b,R)$ in the core box and proves the positivity bounds~\eqref{eq:A-core-positive}--\eqref{eq:U-core-positive}.
Proposition~\ref{prop:tail} gives a unique tail without a growing mode in the certified ball for every $(\Ai,\beta,\eta)$ in the tail box and proves the tail positivity bounds~\eqref{eq:A-tail-positive}--\eqref{eq:U-tail-positive}.
Proposition~\ref{prop:matching} gives a parameter tuple for which the core and tail match in five coordinates.
The constraint monotonicity~\eqref{eq:constraint-monotonicity} completes the sixth coordinate.
Lemma~\ref{lem:constraint} then shows that the global matched orbit satisfies all reduced Einstein--Weyl equations.
Equations~\eqref{eq:h-tail},~\eqref{eq:f-tail}, and~\eqref{eq:asymptotic-flatness} prove asymptotic flatness after constant time normalization.
The positivity estimates throughout all three domains prove~\eqref{eq:metric-positivity-theorem}.
The ingoing extension and causal argument in Section~\ref{sec:future-horizon} identify the future event horizon relative to the chosen asymptotically flat end.
The trace equation gives $R=0$.
Finally, the $b$ interval excludes zero, and $b\neq0$ implies a nonzero Bach tensor at the horizon, so~\eqref{eq:Ricci-Bach} gives $R_{\mu\nu}\neq0$.
\end{proof}

\section{Global geometry and physical parameters}
\label{sec:geometry}

\subsection{The future event horizon and regular exterior}
\label{sec:future-horizon}

The metric functions are controlled in variables adapted to each domain.
In the horizon series domain,
\begin{equation}
 \widehat h=zQW^{2}\,,
 \qquad
 f=zQ\left(\frac{W_{z}}{W}\right)^{2}\,,
 \label{eq:hf-horizon-repeat}
\end{equation}
and the certified inequalities $W>0$, $W_{z}>0$, and $Q>0$ show that the common zero at $z=0$ is simple and isolated.
In the core and tail domains,
\begin{equation}
 \widehat h
 =
 \frac{R^{2}}{A^{2}}
 \left(1+\frac{C}{t}\right)\,,
 \qquad
 f
 =
 \left(1+\frac{C}{t}\right)
 U^{2}\,,
 \label{eq:hf-core-tail}
\end{equation}
with $A>0$, $t+C>0$, and $U>0$.
Thus neither metric function vanishes again.
The areal radius is monotone because
\begin{equation}
 \frac{\dd\rho}{\dd t}
 =
 \frac{U}{A}>0\,.
 \label{eq:rho-monotone}
\end{equation}
The domain of outer communication therefore has the standard topology $(\bar r_{h},\infty)\times\RR\times S^{2}$.
Choose the exterior time orientation so that the static Killing field is future directed.
The chart in~\eqref{eq:dimensionless-kundt} has the outgoing Eddington--Finkelstein sign and directly extends through the past horizon.
To construct the future extension, write $\chi(z)=zQ(z)$ and introduce
\begin{equation}
 r_* =\int^z\frac{\dd\zeta}{\zeta Q(\zeta)}
      =\log|z|+\psi(z)\,,\qquad
 v=\upsilon+2r_*\,.
 \label{eq:future-horizon-tortoise}
\end{equation}
Since $Q(0)=1$, the derivative $\psi'(z)=[Q(z)^{-1}-1]/z$ is analytic at zero.
On the exterior overlap the metric becomes
\begin{equation}
 m_2^2\dd s^2
 =W^2\left[\dd\Omega_2^2-\chi(z)\dd v^2+2\dd v\,\dd z\right]\,.
 \label{eq:ingoing-horizon-metric}
\end{equation}
Its coefficients extend analytically through $z=0$ and the metric is nondegenerate there.
Choose $\epsilon>0$ small enough that $W,Q>0$ for $|z|<\epsilon$, and adjoin this ingoing collar to the entire certified exterior.
This defines a vacuum spacetime $\mathcal M_f$ with coordinates $v\in\RR$, $-\epsilon<z<R$, and $S^2$, where the equations in the collar follow by analytic continuation of the horizon series.

Define the radial null fields
\begin{equation}
 L=\partial_v+\frac{\chi}{2}\partial_z\,,\qquad
 N=-\partial_z\,,\qquad
 g(L,N)=-\frac{W^2}{m_2^2}<0\,.
 \label{eq:future-null-basis}
\end{equation}
They are future directed on the exterior and fix the continued time orientation.
Every future causal tangent has the form $X=c_L L+c_N N+S$, with $c_L,c_N\geq0$ and $S$ tangent to the sphere.
For a nonzero such tangent, causality prevents $c_L=c_N=0$.
Consequently
\begin{equation}
 X(z)=\frac{c_L\chi}{2}-c_N<0\quad(-\epsilon<z<0)\,,\qquad
 X(z)\leq0\quad(z=0)\,.
 \label{eq:future-causal-radial}
\end{equation}
No future causal curve from the interior collar can cross $z=0$ into the exterior.
The inequality $\dd z\leq\chi(z)\dd v/2$, with $\dd v\geq0$ and $\chi(z)=z+O(z^2)$, also excludes outward escape from the horizon itself by comparison with the unique radial null solution $z=0$.
Conversely, every exterior point admits an outgoing radial null ray escaping to the chosen asymptotically flat end.
Its areal radius grows without bound and its affine parameter is unbounded, since in the asymptotically normalized metric $\dd\bar r/\dd\lambda=E\sqrt{f/h}\to E>0$ for positive conserved Killing energy $E$.
Writing $\mathcal I^+$ for future null infinity of that end, these facts give
\begin{equation}
 J^-(\mathcal I^+)\cap\mathcal M_f=\{z>0\}\,,\qquad
 \partial_{\mathcal M_f}J^-(\mathcal I^+)=\{z=0\}\,.
 \label{eq:future-event-horizon}
\end{equation}
Thus $z=0$ is a future event horizon in this explicitly specified extension.
This conclusion requires neither a maximal extension nor a description of the geometry beyond the interior collar.

The nonzero Ricci tensor can also be read directly at the horizon.
In physical areal radius,
\begin{equation}
 R^\theta{}_\theta=
 \frac{1-f-\bar r f_{\bar r}/2-\bar r f\widehat h_{\bar r}/(2\widehat h)}{\bar r^2}\,.
 \label{eq:angular-ricci}
\end{equation}
The simple zero expansions give $f_{\bar r}(\bar r_h)=(1+b)/\bar r_h$ and hence $R^\theta{}_\theta|_h=R^\phi{}_\phi|_h=-b/\bar r_h^2$.
Regularity at the static horizon gives equality of the two normal eigenvalues; together with $R=0$ this yields
\begin{equation}
 R^\tau{}_\tau|_h=R^{\bar r}{}_{\bar r}|_h=\frac{b}{\bar r_h^2}\,,\qquad
 R_{\mu\nu}R^{\mu\nu}|_h=\frac{4b^2}{\bar r_h^4}\,.
 \label{eq:horizon-ricci}
\end{equation}
The certified interval excludes $b=0$, so the Ricci and Bach tensors are nonzero.

\subsection{Mass, temperature, and entropy}

The certified tail parameter box gives the mass interval
\begin{equation}
 \begin{gathered}
 \CertifiedMassLower<m_2M\,,\\
 m_2M<\CertifiedMassUpper\,.
 \end{gathered}
 \label{eq:mass-interval}
\end{equation}
At $\rho_h=1$, the horizon expansions give $\widehat h_\rho(1)=1/(1+b)$ and $f_\rho(1)=1+b$.
Thus the surface gravity in the horizon time normalization satisfies
\begin{equation}
 \kappa_{\rm sg}^2=\frac{m_2^2}{4}\widehat h_\rho(1)f_\rho(1)=\frac{m_2^2}{4}\,.
 \label{eq:surface-gravity}
\end{equation}
The asymptotic normalization of time multiplies the temperature by $\Ai/R$, so
\begin{equation}
 \frac{T}{m_{2}}
 =
 \frac{\Ai}{4\pi R}\,.
 \label{eq:temperature-formula}
\end{equation}
The parameter boxes imply
\begin{equation}
 \begin{gathered}
 \CertifiedTemperatureLower<\frac{T}{m_2}\,,\\
 \frac{T}{m_2}<\CertifiedTemperatureUpper\,.
 \end{gathered}
 \label{eq:temperature-interval}
\end{equation}

The absolute Wald entropy depends on the topological term retained in the action~\cite{Wald:1993nt,Iyer:1994ys,Lu:2015cqa}.
For~\eqref{eq:action} exactly as written, the curvature derivative in the Noether charge formula gives
\begin{equation}
 S_{\rm lit}=\frac{\mathcal A_h}{4G}
 +\frac{k}{4G}\int_h C^{abcd}\epsilon_{ab}\epsilon_{cd}\,\dd\mathcal A\,,
 \qquad \epsilon_{ab}\epsilon^{ab}=-2\,.
 \label{eq:wald-integral}
\end{equation}
The horizon expansion has $\mathcal H''(r_h)+2=6(1+b)$, so the normal Weyl contraction is $C^{abcd}\epsilon_{ab}\epsilon_{cd}=-4(1+b)/\bar r_h^2$.
Consequently
\begin{equation}
 S_{\rm lit}=\frac{\pi\bar r_h^2}{G}-\frac{4\pi k}{G}(b+1)
 =\frac{\pi}{Gm_2^2}\bigl[\rho_h^2-2(b+1)\bigr]\,.
 \label{eq:wald-entropy}
\end{equation}
Here ``lit'' denotes the literal action~\eqref{eq:action}.
The often used alternative convention adds $+kE_4$ to its Lagrangian, where
\begin{equation}
 E_4=R_{abcd}R^{abcd}-4R_{ab}R^{ab}+R^2\,,\qquad
 C^2=E_4+2R_{ab}R^{ab}-\tfrac23 R^2\,.
 \label{eq:euler-density}
\end{equation}
This leaves the four-dimensional vacuum equations unchanged but shifts the entropy of a spherical horizon by $4\pi k/G$:
\begin{equation}
 S_{\rm shift}=S_{\rm lit}+\frac{4\pi k}{G}
 =\frac{\pi}{Gm_2^2}(\rho_h^2-2b)\,.
 \label{eq:entropy-dimensionless}
\end{equation}
It is this shifted convention that reproduces the area law for Schwarzschild and the entropy quoted in~\cite{Lu:2015cqa,Podolsky:2018pfe,Podolsky:2019gro}.
For the certified solution, $Gm_2^2S_{\rm lit}\simeq-5.42428567946$ and $Gm_2^2S_{\rm shift}\simeq0.858899627724$; rigorous outward intervals are supplied in Appendix~\ref{app:poincare} and the ancillary physical parameter output.
The sign of an absolute entropy with this additive topological offset does not establish dynamical stability.

For a direct numerical comparison,~\cite{Podolsky:2018pfe}, Figures~1--2, uses $r_h=-1$, $k=1/2$, hence the same $\rho_h=1$ and the same definition of $b$.

\begin{table}[ht]
\centering
\begin{tabular}{@{}lll@{}}
\toprule
Quantity & Ref.~\cite{Podolsky:2018pfe} & Present certified solution (rounded)\\
\midrule
$b$ & $0.3633018769168$ & $0.36330187862793$\\
$2m_2M$ & $\simeq0.55$ & $0.55107010500725$\\
\bottomrule
\end{tabular}
\caption{Comparison in a common normalization.
The final column displays approximations to enclosures, not new interval endpoints.}
\label{tab:literature-comparison}
\end{table}

The difference in $b$ is about $1.71113\times10^{-9}$, or $4.71\times10^{-9}$ relative: the values agree to eight significant figures, not to every printed digit of the earlier shooting value.
This comparison supports identification with the known numerical branch; it is neither an input to the proof nor a proof that all numerical or exact solutions at this radius belong to a single branch.

\subsection{Scope and possible extensions}
\label{sec:scope}

The theorem gives a globally defined constructive solution even though no finite elementary expression for $h$ and $f$ is known.
The metric is determined by a convergent recurrence, a validated compact core orbit, and a validated asymptotic fixed point, and every global matching inequality required for existence is certified.

The Poincar\'e--Miranda argument proves that at least one parameter tuple lies in the five-dimensional box.
It does not exclude two or more zeros inside that box.
A possible route to uniqueness is an interval Newton or Krawczyk inclusion for the full matching map.
The core first and second variations are already enclosed, but derivatives of the approximate tail shape are not derivatives of the exact fixed point.
For the centered operator $w=\cT_\lambda(w)$, $\lambda=(\Ai,\beta,\eta)$, its parameter derivative would be enclosed through
\begin{equation}
 (I-D_w\cT_\lambda)D_\lambda w=D_\lambda\cT_\lambda\,.
 \label{eq:tail-parameter-derivative}
\end{equation}
The uniform contraction gives $\|(I-D_w\cT_\lambda)^{-1}\|<2.40$.
An enclosure of the right hand side, followed by an interval inclusion for the resulting $5\times5$ Jacobian, is therefore a plausible additional calculation; it has not been performed here.
Small parameter radii alone do not establish uniqueness.
The tail fixed point itself is unique in the specified weighted ball for each fixed parameter tuple, and the core initial value orbit is unique for each fixed $(b,R)$, but these facts do not by themselves imply uniqueness of the global parameter match.

The numerical certificate fixes $\rho_h=1$, but the strict bounds imply qualitative local continuation without a nondegeneracy assumption on the matching Jacobian, as shown below.
An invertible exact matching Jacobian would additionally give a locally unique smooth parameterization by the horizon radius on both sides of the certified point.
At the Schwarzschild bifurcation, one should instead combine the present estimates with a validated Lyapunov--Schmidt reduction of the Lichnerowicz zero mode~\cite{Lu:2017kzi}.

\subsection{A local continuous family at fixed coupling}
\label{sec:continuous-family}

The strict matching certificate also gives a local family of solutions in the same physical theory.
The numerical bounds remain those at $\rho_h=1$; the continuation argument is qualitative.

\begin{proposition}[Local continuous family at fixed coupling]
\label{prop:continuous-family}
Fix $k>0$ and $m_2^2=1/(2k)$.
There exists $\varepsilon>0$ such that, for every $q=m_2\bar r_h\in(1-\varepsilon,1+\varepsilon)$, there is a static, spherically symmetric, asymptotically flat, non-Schwarzschild black hole with a regular nondegenerate horizon and no additional exterior horizon.
Moreover, for some $\delta>0$, these solutions admit a continuous selection $g_q$, $q\in[1,1+\delta)$, whose member $g_1$ has a matching tuple in the certified box of Table~\ref{tab:parameter-box}.
After identifying the exteriors in common regular coordinates and normalizing time at infinity, the family is continuous in $C^\infty_{\mathrm{loc}}$ and real analytic in $q$ for $q>1$.
Different values of $q$ give inequivalent black holes.
\end{proposition}

\begin{proof}
Keep the physical coupling $k$ fixed and use the general horizon normalization of Appendix~\ref{app:recurrence}:
\begin{equation}
 q=m_2\bar r_h\,,\qquad
 W=\frac{\Omega}{\bar r_h}\,,\qquad
 \rho=qW\,,\qquad
 \kappa(q)=\frac{k}{\bar r_h^2}=\frac{1}{2q^2}\,.
 \label{eq:family-normalization}
\end{equation}
Here $W(0)=1$, so varying $q$ changes the horizon radius within the same theory.
Let $K$ denote the closed box of Table~\ref{tab:parameter-box}, with coordinates $\theta=(b,R,\Ai,\beta,\eta)$.
We first extend the exact matching map~\eqref{eq:matching-map} to a real analytic map $\Phi(\theta,q)$ on an open neighborhood of $K\times\{1\}$.

For the horizon solution, write $W=\sum w_nz^n$ and $Q=\sum c_nz^n$, using $c_n$ for the coefficients denoted by $q_n$ in Appendix~\ref{app:recurrence}.
The general recurrence gives
\begin{equation}
 w_0=c_0=1\,,\qquad w_1=1+b\,,\qquad
 c_1=-2-3b\,,\qquad c_2=1+(4-q^2)b+3b^2\,.
 \label{eq:family-horizon-data}
\end{equation}
The recurrence~\eqref{eq:direct-w-recurrence} is unchanged, while the factor $6$ in~\eqref{eq:direct-q-recurrence} becomes $6q^2$.
Consequently the $Q$ majorant ratio is multiplied by $|q|^2$ and the $W$ majorant ratio is unchanged.
The uniform strict bounds~\eqref{eq:majorant-finite-maxima} and the all index estimates following them, together with the initial coefficient bounds, persist for complex $q$ sufficiently close to one.
Thus the horizon series converges normally on a common disk and depends analytically on $(b,q)$.
The compact core depends analytically on $(b,R,q)$ by analytic dependence of the ODE flow, since its denominators remain uniformly separated from zero on the certified compact interval.

For the tail, retain the fixed approximation $X_0$, exponents $a_0,\nu_0$, Green kernels, and weighted Banach space of Section~\ref{subsec:tail-proof}.
The dependence on $q$ enters through $\kappa(q)$; in particular, the nonlinear term used in the centered operator becomes
\begin{equation}
 H_q=\frac{3t}{\kappa(q)A^3}F-a_0^2Y\,.
 \label{eq:family-tail-H}
\end{equation}
This is a perturbation of the original operator in the same space, without changing the approximation exponents to the parameter dependent physical exponents.
The centered operator $\cT(Z;\lambda,q)$, with $\lambda=(\Ai,\beta,\eta)$, is analytic on an open Banach-space neighborhood of the certified set.
Indeed, the denominator bounds permit uniformly convergent geometric expansions, and the integral operators~\eqref{eq:tail-integrals-A}--\eqref{eq:tail-integral-v} are bounded in the norm~\eqref{eq:tail-norm}.
The factors of $t$ multiplying quadratic decaying terms are harmless because
\begin{equation}
 \sup_{t\geq T}t\,\e^{-\sigma(t-T)}<\infty\,.
 \label{eq:family-weight-bound}
\end{equation}
The strict contraction and self-mapping bounds~\eqref{eq:tail-contraction} and~\eqref{eq:tail-self-map} persist under sufficiently small parameter perturbations.
Since $I-D_Z\cT$ is invertible, the analytic implicit function theorem applied to $Z-\cT(Z;\lambda,q)$ gives analytic dependence of the exact tail fixed point on $(\lambda,q)$.
This uses the established tail contraction and does not assume invertibility of the matching Jacobian $D_\theta\Phi$.
Compactness of $K$ makes the preceding neighborhoods uniform, proving the asserted analytic extension of $\Phi$.

Let $G(\theta,q)=P\Phi(\theta,q)$, with the same fixed preconditioner as in~\eqref{eq:preconditioned-map}.
On the opposite faces $K_i^\pm$ of $K$, the certificate gives
\begin{equation}
 G_i(\theta,1)<0\quad(\theta\in K_i^-)\,,\qquad
 G_i(\theta,1)>0\quad(\theta\in K_i^+)\,.
 \label{eq:family-face-signs}
\end{equation}
The face margins are uniformly positive.
Joint continuity and compactness preserve these signs for $|q-1|<\varepsilon$, after reducing $\varepsilon$ if necessary.
The Poincar\'e--Miranda theorem therefore supplies a zero of $\Phi(\,\cdot\,,q)$ in $\operatorname{int}K$ for every such $q$.
The strict constraint monotonicity~\eqref{eq:constraint-monotonicity} also persists.
Both the horizon solution and the tail obey the constraint with coupling $\kappa(q)$, so the five coordinate match again completes to a six state solution.

To obtain a continuous selection, choose $q_n\downarrow1$ and matching zeros $\theta_n\in K$.
A subsequence converges to $\theta_*$ with $\Phi(\theta_*,1)=0$.
The strict face signs imply $\theta_*\in\operatorname{int}K$.
The set
\begin{equation}
 \mathcal S=
 \left\{(\theta,q)\in\operatorname{int}K\times(1,1+\varepsilon):
 \Phi(\theta,q)=0\right\}
 \label{eq:family-semianalytic-set}
\end{equation}
is semianalytic, and $(\theta_*,1)$ lies in its closure.
The real analytic curve selection lemma~\cite[Section~7, statement~(7.2)]{Ruiz:1986} gives an analytic arc $s\mapsto(\vartheta(s),q(s))$ starting at $(\theta_*,1)$ and lying in $\mathcal S$ for $s>0$.
Since $q(s)>1$, there are an integer $m\geq1$ and a constant $c>0$ such that
\begin{equation}
 q(s)=1+cs^m+O(s^{m+1})\,.
 \label{eq:family-arc}
\end{equation}
Thus $q(s)$ is strictly increasing for sufficiently small positive $s$.
Reparameterization by $q$ gives a continuous selection $\theta(q)$ on $[1,1+\delta)$, analytic for $q>1$.

The geometric positivity bounds persist in all three domains after reducing the neighborhood.
The regular ingoing extension and causal argument of Section~\ref{sec:future-horizon} then identify the future event horizon relative to the chosen asymptotically flat end.
The interval for $b$ excludes zero, so all selected solutions remain non-Schwarzschild.
Use $x=z/R$ and a regular ingoing time coordinate to identify the exteriors, and normalize time at infinity.
The analytic horizon series and ODE flow give the stated dependence of the metric on compact subsets, including a common horizon collar.
At infinity, the asymptotic parameters vary continuously, and the decaying state $X=(A-\Ai,p,C-\beta,d,u_0,v_0)$ varies continuously in the norm~\eqref{eq:tail-norm}.
Finally, the horizon area is
\begin{equation}
 \mathcal A_h(q)=\frac{4\pi q^2}{m_2^2}\,,
 \label{eq:family-horizon-area}
\end{equation}
so different values of $q$ give inequivalent black holes.
\end{proof}

The proposition requires no new numerical continuation calculation, but it gives no numerical lower bound on $\varepsilon$ or $\delta$.
It does not establish uniqueness, differentiability of the radius parameterization at $q=1$, or continuation through every matching zero in the certified box.
The analogous construction for $q<1$ may approach a different matching zero, so no two-sided selection through one specified zero is asserted.

\section{The cold endpoint}
\label{sec:cold-endpoint}

The quantitative certificate concerns a finite temperature solution at $\rho_h=1$.
Neither this certificate nor the qualitative local family of Proposition~\ref{prop:continuous-family} establishes continuation towards zero horizon radius or zero temperature.
Reference~\cite{Goldstein:2017rxu} discusses a possible cold limit using numerical extrapolation, with increasing integration difficulties.
The relation of the asymptotic Yukawa amplitude to thermodynamics, the broader Einstein--Weyl phase diagram, and late stage evaporation scenarios are investigated in~\cite{Bonanno:2019,Silveravalle:2023,Bonanno:2024}.
These studies motivate a separate global analysis; no quantitative continuation data towards a cold endpoint are asserted here.

One local obstruction can be proved without such continuation.
\begin{proposition}[No regular extremal horizon of finite area]
\label{prop:no-extremal}
A static spherical vacuum solution of~\eqref{eq:action} cannot have a smooth extremal horizon represented in a regular Kundt chart by $0<|\Omega_h|<\infty$, finite derivatives through the order of the field equations, and $\mathcal H_h=\mathcal H'_h=0$.
\end{proposition}
\begin{proof}
The trace~\eqref{eq:kundt-trace} at the horizon gives $\mathcal H''_h=-2$.
Equation~\eqref{eq:kundt-two} then reduces to
\begin{equation}
 \Omega_h^2=\frac{k}{3}\left[-\frac12(-2)^2+2\right]=0\,,
\end{equation}
contradicting finite nonzero horizon area.
\end{proof}

If a continuation were established with $\rho_h\to0$ and $b\to-1$,~\eqref{eq:horizon-ricci} would imply divergent horizon curvature, $R_{\mu\nu}R^{\mu\nu}|_h\sim4m_2^4/\rho_h^4$.
Those horizon limits alone do not determine the temperature, which also depends on the normalization of time at infinity.
The same conditional limit gives $S_{\rm lit}\to0$ for the stated action and $S_{\rm shift}\to2\pi/(Gm_2^2)$ after adding the Euler density.
Thus a finite nonzero entropy offset in this scenario is convention dependent and does not establish a regular remnant or the existence of a cold endpoint.

\section{Discussion}
\label{sec:discussion}

The non-Schwarzschild Einstein--Weyl black hole has long occupied an intermediate status.
Its numerical construction is robust, and several independent representations reproduce its geometry and thermodynamics, but a global analytic existence statement was absent~\cite{Lu:2015cqa,Lu:2015psa,Goldstein:2017rxu,Kokkotas:2017zwt,Podolsky:2019gro}.
The obstruction is not a failure of local analyticity at the horizon.
It is the exponentially unstable connection problem between a regular horizon and an irregular asymptotic endpoint.

The proof given here treats that connection problem directly.
The horizon recurrence supplies analytic local data.
The compact core is enclosed without losing the tiny two parameter correlations that select the asymptotically flat orbit.
At infinity, the growing massive mode is removed as a boundary condition in a weighted integral equation rather than by unstable outward shooting.
Finally, the parameter match is obtained topologically, so a full interval enclosure of the exact $5\times5$ derivative is unnecessary for existence.

The certified solution also gives rise to a local continuous family at fixed Einstein--Weyl coupling.
Proposition~\ref{prop:continuous-family} makes this precise.
The strict Poincar\'e--Miranda face inequalities persist when the dimensionless horizon radius varies sufficiently little, giving existence at every nearby radius.
Analytic dependence of the horizon, core, and exact tail constructions then permits application of the real analytic curve selection lemma, yielding a continuous family through at least one solution in the certified box.
This argument does not require nonsingularity of the matching Jacobian.
The family can be parameterized continuously by the horizon radius on one side of the certified point and is analytic away from that endpoint.
This qualitative local conclusion does not establish uniqueness, a smooth two-sided radius parameterization through the certified point, or validated continuation over a specified numerical interval.

The proof also clarifies why the algebraic asymptotic expansion does not reveal the non-Schwarzschild hair.
The formal generalized Yukawa sector has $A-\Ai\sim\e^{-at}t^{\nu-1}$ and $C-\beta\sim\e^{-at}t^\nu$.
For the metric functions this gives the distinct formal leading powers
\begin{equation}
 h-h_{\rm Sch}\sim c_h\rho^{-1-\mu}\e^{-\rho}\,,\qquad
 f-f_{\rm Sch}\sim c_f\rho^{-\mu}\e^{-\rho}\,,
 \label{eq:Yukawa-summary}
\end{equation}
where $h_{\rm Sch}=f_{\rm Sch}=1-2\mu/\rho$ and the constants depend on normalization.
The additional power in $f$ comes from the factor $tp$ in~\eqref{eq:f-tail}.
The existence theorem uses the weaker rigorous exponential envelope in~\eqref{eq:asymptotic-flatness}; it does not prove the sharp asymptotic equivalences in~\eqref{eq:Yukawa-summary}.
Both rates are beyond every algebraic order in $1/\rho$.

The theorem establishes a classical configuration with massive spin-two hair and does not include a perturbative stability proof.
Existing numerical work finds a monopole instability below the branch intersection near $\rho_h\simeq0.87$ and no such instability above it~\cite{Held:2023}.
The certified radius $\rho_h=1$ lies above that threshold.
More recent quasinormal mode calculations find numerical evidence for axial stability and radially stable parameter regions~\cite{Antoniou:2025}.
These mode calculations provide context, but do not constitute a complete rigorous stability analysis of the certified background.

Several extensions would require further validation.
One can obtain explicit enclosures on a specified horizon-radius interval and its mass versus temperature curve, prove local uniqueness by an interval Newton argument, or combine the spherical existence theorem with a spectral stability calculation.
The same domain decomposition strategy may also be useful for black holes with asymptotically AdS geometry in higher derivative theories~\cite{Svarc:2018} and for matter coupled systems with one or more exponentially growing modes.

\acknowledgments

The authors thank the developers of MPFR and the broader validated numerics community whose software and methods make computer assisted proofs of nonlinear boundary value problems possible.
The computations reported here were performed with directed rounding interval arithmetic and exact rational postprocessing.
VJ thanks the NSF Institute for Artificial Intelligence and Fundamental Interactions (IAIFI) and the Department of Physics at Northeastern University for hospitality during his sabbatical during which much of this research was undertaken.
VJ is supported by the South African Research Chairs Initiative of the Department of Science, Technology, and Innovation and the National Research Foundation (grant 78554).

\appendix

\section{Horizon recurrence and convergence}
\label{app:recurrence}

\subsection{Direct recurrence and the horizon constraint}

In the dimensionless horizon gauge the reduced equations have $\Omega=W$, $\mathcal H=-zQ$, and $\kappa=1/2$.
Thus the first equation and the trace equation are
\begin{align}
 WW_{zz}-2W_z^2&=-\frac16(zQ)_{zzzz}\,,
 \label{eq:horizon-direct-first}\\
 zQW_{zz}+(Q+zQ_z)W_z+
 \frac16(2Q_z+zQ_{zz}-2)W&=0\,.
 \label{eq:horizon-direct-trace}
\end{align}
Write $W=\sum_{n\ge0}w_nz^n$ and $Q=\sum_{n\ge0}q_nz^n$.
The initial data are
\begin{equation}
 w_0=q_0=1\,,\quad w_1=1+b\,,\quad q_1=-2-3b\,,\quad
 q_2=1+3b+3b^2 \,.
 \label{eq:direct-initial}
\end{equation}
Equating coefficients gives, for $n\ge1$,
\begin{align}
 (n+1)^2w_{n+1}
 ={}&\frac{w_n}{3}
 -(n+1)\sum_{k=1}^{n}(n+1-k)q_kw_{n+1-k}
 \nonumber\\[-1mm]
 &-\frac16\sum_{k=1}^{n+1}k(k+1)q_kw_{n+1-k}\,,
 \label{eq:direct-w-recurrence}
\end{align}
and, for $n\ge0$,
\begin{align}
 q_{n+3}={}&-\frac{6}{(n+1)(n+2)(n+3)(n+4)}
 \sum_{i=0}^{n}\Big[
 (n-i+2)(n-i+1)w_iw_{n-i+2}
 \nonumber\\[-1mm]
 &\hspace{48mm}
 -2(i+1)(n-i+1)w_{i+1}w_{n-i+1}\Big]\,.
 \label{eq:direct-q-recurrence}
\end{align}
The induction is triangular in the order $w_2,q_3,w_3,q_4,\ldots$: the $w_{n+1}$ step uses $q_{n+1}$, whereas the $q_{n+3}$ step uses $w_{n+2}$.
Consequently no unknown coefficient is used to bound itself.

For completeness, the implementation uses the equivalent Schwarzschild--Bach coefficients
\begin{equation}
 w_0=1\,,\quad w_n=1+b\alpha_n\ (n\ge1)\,,\qquad
 Q=(1-z)^2+3b\sum_{n\ge1}(-1)^n\gamma_nz^n \,.
 \label{eq:wn}
\end{equation}
For a general dimensionless horizon radius $\rho_h$, use the normalized function $W=\Omega/\bar r_h$, $z=1+\bar r_h r$, and the coupling $\kappa=k/\bar r_h^2=1/(2\rho_h^2)$.
In this general-radius normalization, also used in Section~\ref{sec:continuous-family}, $W(0)=1$ denotes normalization to the horizon radius; $m_2\bar r=\rho_h W$.
Their initial values are
\begin{equation}
 \alpha_0=0\,,\quad\alpha_1=\gamma_1=1\,,\qquad
 \gamma_2=\frac{4-\rho_h^2+3b}{3}\,,
 \label{eq:recurrence-initial}
\end{equation}
and, for $j\ge2$,
\begin{align}
 \alpha_j=\frac{1}{j^2}\Big[&
 (2j^2-2j+1)\alpha_{j-1}-(j-1)^2\alpha_{j-2}
 \nonumber\\
 &-3\sum_{i=1}^{j}(-1)^i\gamma_i(1+b\alpha_{j-i})
 \left(j(j-i)+\frac{i(i+1)}{6}\right)\Big]\,,
 \label{eq:alpha-recurrence}\\
 \gamma_{j+1}=\frac{(-1)^j}{\kappa(j+2)(j+1)j(j-1)}
 &\sum_{i=0}^{j-1}
 [\alpha_i+\alpha_{j-i}(1+b\alpha_i)](j-i)(j-1-3i)\,.
 \label{eq:gamma-recurrence}
\end{align}
Only $\rho_h=1$ is used in the certified integration.
At $\rho_h=1$, substitution of~\eqref{eq:wn} into these two recurrences yields~\eqref{eq:direct-w-recurrence}--\eqref{eq:direct-q-recurrence}.

The remaining field equation is also satisfied.
With derivatives in $z$, define its residual by
\begin{equation}
 \mathcal E_h=
 WW_z\mathcal H_z+3W_z^2\mathcal H+W^2
 -\frac{\kappa}{3}\left(\mathcal H_z\mathcal H_{zzz}
       -\frac12\mathcal H_{zz}^2+2\right)\,.
 \label{eq:horizon-constraint-residual}
\end{equation}
The first equation gives the identity
\begin{equation}
 \frac{d\mathcal E_h}{dz}
 =6W_z\left[\mathcal H W_{zz}+\mathcal H_zW_z+
                  \frac{1}{6}(\mathcal H_{zz}+2)W\right]=0\,.
 \label{eq:horizon-constraint-identity}
\end{equation}
At $z=0$ the left side of the constraint equation is $-b$, while its right side in the general-radius normalization is $2\kappa(q_2-q_1^2/3+1/3)=-2\kappa\rho_h^2b=-b$.
Thus~\eqref{eq:recurrence-initial} sets $\mathcal E_h(0)=0$ for general $\rho_h$, reducing to~\eqref{eq:direct-initial} when $\rho_h=1$.
After convergence is established below, the residual is identically zero on the analytic horizon solution and therefore at the series--core join.
This argument uses the regular horizon variables, rather than the singular $t=1$ form of the core equations.

\subsection{Explicit majorant induction}

Set $L=10$, $B=5$ and
\begin{equation}
 U_n=\frac{L^n}{(n+1)^2}\,,\qquad V_n=\frac{BL^n}{(n+1)^4}\,.
 \label{eq:majorant-repeat}
\end{equation}
For $|b|\le37/100$ the five initial bounds follow from
\begin{equation}
 |w_0|=1\,,\quad |w_1|\le\frac{137}{100}\le\frac{10}{4}\,,\quad
 |q_0|=1\,,\quad |q_1|\le\frac{311}{100}\le\frac{50}{16}\,,\quad
 |q_2|\le\frac{25207}{10000}\le\frac{500}{81}\,.
 \label{eq:majorant-initial-check}
\end{equation}
Assuming the preceding bounds in the triangular order, the $q_{n+3}$ step has ratio at most
\begin{align}
 \frac{|q_{n+3}|}{V_{n+3}}
 &\le r_q(n):=\frac{6}{BL}\,
 \frac{(n+4)^3}{(n+1)(n+2)(n+3)}\,S_n\,,
 \label{eq:q-majorant-ratio}\\
 S_n&=\sum_{i=0}^{n}\left[
 \frac{(n-i+2)(n-i+1)}{(i+1)^2(n-i+3)^2}
 +\frac{2(i+1)(n-i+1)}{(i+2)^2(n-i+2)^2}\right]\,.
 \nonumber
\end{align}
For the $w$ step it is useful to keep the five convolution terms of~\eqref{eq:horizon-direct-trace} separate.
With $\phi_n=(n+2)^2/(n+1)^2$,
\begin{equation}
 \frac{|w_{n+1}|}{U_{n+1}}
 \le r_w(n):=\phi_n(A_n+B_n+C_n+D_n+E_n+F_n)\,,
 \label{eq:w-majorant-ratio}
\end{equation}
where empty sums vanish and
\begin{align}
 A_n&=5\sum_{i=1}^{n-1}
 \frac{(n+1-i)(n-i)}{(i+1)^4(n+2-i)^2}\,,
 &
 B_n&=5\sum_{i=1}^{n}
 \frac{n+1-i}{(i+1)^4(n+2-i)^2}\,,
 \nonumber\\
 C_n&=5\sum_{i=1}^{n}
 \frac{i(n+1-i)}{(i+1)^4(n+2-i)^2}\,,
 &
 D_n&=\frac{1}{30(n+1)^2}\,,
 \nonumber\\
 E_n&=\frac{5}{3}\sum_{i=0}^{n}
 \frac{i+1}{(i+2)^4(n-i+1)^2}\,,
 &
 F_n&=\frac{5}{6}\sum_{i=0}^{n-1}
 \frac{(i+2)(i+1)}{(i+3)^4(n-i)^2}\,.
 \label{eq:w-six-sums}
\end{align}
The factors of $L$ cancel in the five convolutions; the isolated $W/3$ term gives $D_n$.
Exact rational evaluation gives
\begin{equation}
 \max_{0\le n<100}r_q(n)=\frac{4}{9}\,,\qquad
 \max_{1\le n<100}r_w(n)=\frac{16817}{29160}\,.
 \label{eq:majorant-finite-maxima}
\end{equation}

The following estimates cover all remaining indices.
The integral test gives
\begin{equation}
 \zeta(2)<\frac{5}{3}\,,\qquad
 \zeta(3)-1<\frac{1}{4}\,,\qquad
 \zeta(4)-1<\frac{1}{12}\,.
 \label{eq:rational-zeta-bounds}
\end{equation}
For the first and third inequalities one may sum through $k=5$ and add respectively $1/5$ and $1/(3\,5^3)$; for the middle one, $1/2^3+\int_2^\infty x^{-3}\,dx=1/4$.
The first summand of $S_n$ has sum below $5/3$.
Writing $k=i+2$ in the second gives
\begin{equation}
 S_n<\frac{5}{3}+\frac{4(H_{n+2}-1)}{n+4}
 \le \frac{5}{3}+\frac{2}{5}=\frac{31}{15}\,,\qquad n\ge100\,.
 \label{eq:q-high-convolution-bound}
\end{equation}
The last inequality can be checked without transcendental arithmetic: $(H_{102}-1)/104<1/10$ by an exact rational sum, and $(H_{n+2}-1)/(n+4)$ decreases thereafter, since $H_{n+2}-1\ge13/12>(n+4)/(n+3)$.

For the six $w$ sums, the required bounds are
\begin{align}
 A_n&<\frac{5}{12}\,,
 &B_n&\le \frac{5}{6(n+4)}+\frac{8}{(n+2)^3}\,,
 \nonumber\\
 C_n&\le \frac{5}{2(n+4)}+\frac{4}{(n+2)^2}\,,
 &D_n&=\frac{1}{30(n+1)^2}\,,
 \nonumber\\
 E_n&\le \frac{5}{3(n+2)^2}+\frac{200}{9(n+2)^3}\,,
 &F_n&\le \frac{100}{9(n+1)^2}\,.
 \label{eq:w-high-convolution-bounds}
\end{align}
For $A_n$, bound its numerator by $(n+2-i)^2$ and sum $5/(i+1)^4$.
For $B_n$ and $C_n$, use $(j+1)/(j+2)^2\le1/(j+3)$ and put $k=i+1$, $M=n+4$.
The resulting sums are bounded by $5\sum k^{-4}(M-k)^{-1}$ and $5\sum k^{-3}(M-k)^{-1}$ respectively, with $k\ge2$ and $M-k\ge3$.
On $k\le M/2$, use $(M-k)^{-1}\le2/M$ and~\eqref{eq:rational-zeta-bounds}.
On $k>M/2$, the integral test and $\lfloor M/2\rfloor\ge(M-2)/2$ give the respective bounds $40/[9(M-2)^3]\le8/(M-2)^3$ and $10/[3(M-2)^2]\le4/(M-2)^2$.

For $E_n$, put $k=i+2$, $l=n-i+1$, so $k+l=n+3$, and bound $(k-1)/k^4$ by $k^{-3}$.
On $k\le(n+2)/2$, use $l^{-2}\le4/(n+2)^2$ and $\sum_{k\ge2}k^{-3}<1/4$; on the complementary part use $k^{-3}\le8/(n+2)^3$ and $\sum_{l\ge1}l^{-2}<5/3$.
For $F_n$, put $k=i+3$, $l=n-i$ and bound $(k-1)(k-2)/k^4$ by $k^{-2}$.
Splitting the sum at $(n+3)/2$ gives $F_n<100/[9(n+3)^2]\le100/[9(n+1)^2]$.

Substituting~\eqref{eq:q-high-convolution-bound} and~\eqref{eq:w-high-convolution-bounds} at $n=100$ yields $r_q(n)<0.263$ and $r_w(n)<0.460$ for every $n\ge100$.
Indeed, each term of the $w$ upper bound and $\phi_n$ decreases, as does the $q$ prefactor $\prod_{j=1}^{3}(n+4)/(n+j)$.
Together with~\eqref{eq:majorant-finite-maxima}, this completes the induction and proves Proposition~\ref{prop:horizon-majorant}.
The coefficient bounds imply normal convergence for $|z|<1/10$, including every fixed derivative on a smaller disk, so termwise substitution into the equations and the constraint identity is justified.

\subsection{Differentiated and parameter dependent remainders}
\label{app:horizon-remainders}

For $j\ge0$, define the positive majorant tails
\begin{equation}
 T^W_j(N,z_*)=\sum_{n>N}\frac{n^{\underline j}10^n z_*^{\,n-j}}{(n+1)^2}\,,
 \qquad
 T^Q_j(N,z_*)=5\sum_{n>N}\frac{n^{\underline j}10^n z_*^{\,n-j}}{(n+1)^4}\,.
 \label{eq:general-horizon-tail}
\end{equation}
They bound the omitted $j$th derivatives uniformly on $|z|\le z_*$.
For $N=180$, $z_*=7/200$, and $j\le5$, the ratio of successive terms is bounded by
\begin{equation}
 \frac{7}{20}\frac{n+1}{n+1-j}
       \left(\frac{n+1}{n+2}\right)^p<\frac{9}{25}\,,
 \qquad n\ge181\,,\quad p=2\ \hbox{or}\ 4 \,.
 \label{eq:horizon-tail-ratio}
\end{equation}
A finite sum followed by a geometric tail therefore gives exact rational bounds for all the jets used in the initialization.

The coefficients are polynomials in $b$.
The complex disk of radius $\delta_b=0.0066$ about every real $b$ in the parameter box lies inside $|b|<0.37$.
Cauchy's inequalities bound the first and second $b$ derivatives of each remainder by $T_j/\delta_b$ and $2T_j/\delta_b^2$.
At the join, $z=R/20$, so differentiation in the scaled parameters $\xi_b,\xi_R$ multiplies these bounds by $s_b=10^{-41}$ and $s_R/20=10^{-27}/20$, respectively.
For every $W$ or $Q$ jet of orders $0,\ldots,3$, a common componentwise remainder bound is
\begin{equation}
\begin{array}{c|cccccc}
 &v&\partial_{\xi_b}&\partial_{\xi_R}&
 \partial_{\xi_b}^2&\partial_{\xi_b}\partial_{\xi_R}&
 \partial_{\xi_R}^2\\ \hline
 \text{absolute bound}&
 10^{-75}&10^{-114}&10^{-100}&10^{-152}&10^{-138}&10^{-124}
\end{array}\,.
 \label{eq:horizon-common-jet-tail}
\end{equation}
For example, the six unrounded bounds for a jet of order $j$ are
\begin{equation}
 T_j\,,\quad \frac{s_bT_j}{\delta_b}\,,\quad \frac{s_RT_{j+1}}{20}\,,
 \quad \frac{2s_b^2T_j}{\delta_b^2}\,,\quad
 \frac{s_bs_RT_{j+1}}{20\delta_b}\,,\quad \frac{s_R^2T_{j+2}}{400}\,.
 \label{eq:horizon-cauchy-tail-formula}
\end{equation}
The ancillary rational certificates check~\eqref{eq:horizon-tail-ratio}--\eqref{eq:horizon-cauchy-tail-formula}, the finite majorant inequalities, and the equivalence of the two coefficient recurrences.
The all index estimates themselves are the inequalities derived above.

\section{Centered tail fixed point estimates}
\label{app:tail}

\subsection{The approximation and the centered operator}

Let $\bar A=\Ai+\alpha_0$ and $\bar C=\beta+c_0$ denote the approximate physical fields.
For a polynomial $P(r)$, define
\begin{equation}
 \mathcal D P=(-a_0+\nu_0r)P-r^2P_r\,,
 \qquad (g(t)P(1/t))'=g(t)(\mathcal DP)(1/t)\,.
 \label{eq:tail-polynomial-derivative}
\end{equation}
The ancillary file contains finite rational coefficient lists $P_A,P_C$; $P_p=\mathcal DP_A$, $P_d=\mathcal DP_C$, $P_Y=\mathcal D^2P_C$, and $P_J=\mathcal D^3P_C$.
Divide both stored polynomials by $(P_J-a_0P_Y)(1/T)$ before forming these derivative polynomials.
The divisor is nonzero and the resulting normalization is checked exactly.
No convergence claim about the infinite generalized Yukawa expansion is needed.

Set $H=6tF/A^3-a_0^2Y$ and define the differential residuals
\begin{equation}
 R_F=p_0'-F(\bar A,p_0,\bar C,d_0,Y_0;t)\,,\qquad
 R_J=J_0'-\frac{6t}{\bar A^3}F(\bar A,p_0,\bar C,d_0,Y_0;t)\,.
 \label{eq:tail-residual-definition}
\end{equation}
The four kinematic residuals vanish identically.
For $Z=(z_\alpha,z_p,z_c,z_d,z_u,z_v)$ let $\delta Y=(z_u-z_v)/(2a_0)$, and let $\Delta F$, $\Delta H$ denote the changes of $F,H$ from the approximation to the corrected fields.
Writing $Q_F=\Delta F-R_F$ and $Q_H=\Delta H-R_J$, the centered operator is
\begin{align}
 (\mathcal TZ)_\alpha(t)&=\int_t^\infty(s-t)Q_F(s)\,\dd s\,,&
 (\mathcal TZ)_p(t)&=-\int_t^\infty Q_F(s)\,\dd s\,,
 \label{eq:tail-integrals-A}\\
 (\mathcal TZ)_c(t)&=\int_t^\infty(s-t)\delta Y(s)\,\dd s\,,&
 (\mathcal TZ)_d(t)&=-\int_t^\infty\delta Y(s)\,\dd s\,,
 \label{eq:tail-integrals-C}\\
 (\mathcal TZ)_u(t)&=-\int_t^\infty\e^{-a_0(s-t)}Q_H(s)\,\dd s\,,
 \label{eq:tail-integral-u}\\
 (\mathcal TZ)_v(t)&=\int_T^t\e^{-a_0(t-s)}Q_H(s)\,\dd s\,.
 \label{eq:tail-integral-v}
\end{align}
The last equation imposes $z_v(T)=0$; the first five impose the required future boundary conditions.
Exponential envelopes make all improper integrals convergent.
Their derivatives recover $\alpha'=p$, $p'=F$, $c'=d$, $d'=Y$ and $u_0'=a_0u_0+H$, $v_0'=-a_0v_0+H$.

\subsection{Residual envelopes on the entire half line}

For $P=\sum_jP_jr^j$ put $[P]_T=\sum_j|P_j|T^{-j}$ and $\bar\eta=|\eta_0|+r_\eta$.
Since $\nu_0<0$, $g(t)\leq\e^{-a_0(t-T)}$.
Define $s_A,s_p,s_c,s_d,s_Y,s_J$ as $\bar\eta[P]_T$ for the six physical shape polynomials.
These are uniform exponential envelope constants.
For this subsection write
\begin{equation}
 A_-^0=\Aizero-r_A-s_A\,,\quad A_+^0=\Aizero+r_A+s_A\,,
 \quad D_0=T+\beta_0-r_\beta\,,\quad D_*=D_0-s_c\,.
 \label{eq:tail-residual-denominators}
\end{equation}
The linearization at constant fields $(A_*,\beta_*)$ is
\begin{equation}
 F_{\rm lin}^{A_*,\beta_*}
 =-\left(\frac1t+\frac1{t+\beta_*}\right)p
 +\frac{A_*Y}{6(t+\beta_*)}
 +\frac{A_*d}{3t(t+\beta_*)}\,.
 \label{eq:tail-linear-F}
\end{equation}
In particular, the derivative with respect to $\beta_*$ of the coefficient of $p$ is $+(t+\beta_*)^{-2}$.

The residual is split into the linear truncation error at $(\Aizero,\beta_0)$, the change of these constant parameters, and the nonlinear remainder.
For the first part define the exact polynomials
\begin{align}
 N_L&=-r(2+\beta_0r)P_p+
       \tfrac13\Aizero r^2P_d+\tfrac16\Aizero rP_Y\,,\\
 N_F&=(1+\beta_0r)\mathcal DP_p-N_L\,,\qquad
 N_J=r(1+\beta_0r)\mathcal DP_J-6\Aizero^{-3}N_L\,.
 \label{eq:tail-residual-polynomials}
\end{align}
Since $N_J(0)=0$, valid envelope constants for the linear residuals are
\begin{equation}
 K_{F,l}=\frac{\bar\eta[N_F]_T}{1+\beta_0/T}\,,\qquad
 K_{J,l}=\frac{\bar\eta\sum_{j\geq1}|(N_J)_j|T^{1-j}}
                       {1+\beta_0/T}\,.
 \label{eq:tail-linear-residual-bounds}
\end{equation}
All coefficients, including the rounding errors of the proposed approximation, enter these exact rational evaluations.

For the parameter and nonlinear parts let $\beta_*=\max(|\beta_0-r_\beta|,|\beta_0+r_\beta|)$ and set
\begin{align}
 R_*&=\frac{r_A}{D_0}
       +\frac{\Aizero r_\beta}{D_0(T+\beta_0)}\,,\\
 V_*&=\frac{s_A}{D_*}
       +\frac{(\Aizero+r_A)s_c}{D_*D_0}\,,\\
 K_{F,p}&=\frac{s_YR_*}{6}+\frac{s_dR_*}{3T}
             +\frac{s_pr_\beta}{D_0^2}\,,\\
 K_{F,n}&=\frac{2s_p^2}{A_-^0}+\frac{s_ps_d}{D_*}
       +\frac{s_ps_c}{D_*D_0}+\frac{s_YV_*}{6}+\frac{s_dV_*}{3T}\,,\\
 K_L&=\frac{(2T+\beta_*)s_p}{TD_0}
       +\frac{(\Aizero+r_A)s_d}{3TD_0}
       +\frac{(\Aizero+r_A)s_Y}{6D_0}\,.
 \label{eq:tail-remaining-residual-bounds}
\end{align}
Here $K_{F,p}$ multiplies $\e^{-a_0(t-T)}$, $K_{F,n}$ multiplies $\e^{-2a_0(t-T)}$, and $K_L$ bounds either relevant linearization with one exponential factor.
Let $I_p$ bound $|\Ai^{-3}-\Aizero^{-3}|$ and let $I_n$ bound $\e^{a_0(t-T)}|\bar A^{-3}-\Ai^{-3}|$.
Their rational values follow by substituting the above endpoint bounds into
\begin{equation}
 |x^{-3}-y^{-3}|=|x-y|\frac{x^2+xy+y^2}{x^3y^3}\,.
 \label{eq:tail-inverse-cube-bound}
\end{equation}
Then
\begin{equation}
 K_{J,p}=6T\left(\frac{K_{F,p}}{(A_-^0)^3}+K_LI_p\right)\,,\qquad
 K_{J,n}=6T\left(\frac{K_{F,n}}{(A_-^0)^3}+K_LI_n\right)\,.
 \label{eq:tail-J-residual-bounds}
\end{equation}
These constants retain the exponential rate $a_0$.
For the parameter terms this follows from their rational coefficients, rather than from an additional exponential factor: $t/(t+\beta_-)$, $t/(t+\beta_-)^2$, and $t/[(t+\beta_-)(t+\beta_0)]$ decrease on $[T,\infty)$, where $\beta_-=\beta_0-r_\beta<0$ and $\beta_-\beta_0<T^2$.
Consequently $t$ times the parameter envelope and $t$ times the linear envelope are bounded by $T$ times their values above.
For the nonlinear terms, $a_0T>1$ gives $t\e^{-2a_0(t-T)}\leq T\e^{-a_0(t-T)}$.
Thus $K_F=K_{F,l}+K_{F,p}+K_{F,n}$ and $K_J=K_{J,l}+K_{J,p}+K_{J,n}$ give~\eqref{eq:RF-bound}--\eqref{eq:RJ-bound}.

\subsection{Derivative bounds and the contraction matrix}

On the correction ball of radius $\varepsilon=5\times10^{-30}$ let
\begin{align}
 A_-&=\Aizero-r_A-s_A-q_\alpha\varepsilon\,,&
 A_+&=\Aizero+r_A+s_A+q_\alpha\varepsilon\,,\\
 C_-&=\beta_0-r_\beta-s_c-q_c\varepsilon\,,&
 C_+&=\beta_0+r_\beta+s_c+q_c\varepsilon\,,\\
 P&=s_p+q_p\varepsilon\,,&D_1&=s_d+q_d\varepsilon\,,\\
 Y_1&=s_Y+s_J/a_0+(q_u+q_v)\varepsilon/(2a_0)\,,\nonumber\\
 D&=T+C_-\,,&\mathcal R&=T/D\,.
 \label{eq:tail-ball-envelopes}
\end{align}
Also write $C_* =\max(|C_-|,|C_+|)$.
The derivative variables are bounded by their indicated constants times $\e^{-\sigma(t-T)}$; the bound $Y_1$ follows conservatively from $|u_0|,|v_0|\leq(s_J+a_0s_Y)\e^{-a_0(t-T)}$ for the approximation.
We have $A_->1$ and $-1<C_-\leq C_+<0$.
The following five entries bound the absolute partial derivatives of $F$ in the physical order $(A,p,C,d,Y)$:
\begin{align}
 f_A&=2P^2/A_-^2+Y_1/(6D)+D_1/(3TD)\,,\\
 f_p&=4P/A_-+D_1/D+C_* /(TD)+2/D\,,\\
 f_C&=A_+Y_1/(6D^2)+A_+D_1/(3TD^2)+PD_1/D^2+P/D^2\,,\\
 f_d&=A_+/(3TD)+P/D\,,\qquad f_Y=A_+/(6D)\,.
 \label{eq:tail-F-derivative-bounds}
\end{align}
For example, $F_C=-AY/(6(t+C)^2)-Ad/(3t(t+C)^2)+pd/(t+C)^2+p/(t+C)^2$; the final term combines the two apparent terms proportional to $p$ whose coefficients depend on $C$ in $F$.
Let
\begin{equation}
 F_* =2P^2/A_-+A_+Y_1/(6D)+A_+D_1/(3TD)
          +PD_1/D+C_*P/(TD)+2P/D\,.
\end{equation}
Using $H_i=6tA^{-3}F_i-18tA^{-4}F\delta_{iA}-a_0^2\delta_{iY}$, valid bounds for the physical partial derivatives of $H$ are
\begin{align}
 h_A&=12TP^2/A_-^5+\mathcal RY_1/A_-^3+2D_1/(A_-^3D)+18TF_*/A_-^4\,,\\
 h_p&=24TP/A_-^4+6\mathcal RD_1/A_-^3+6C_* /(A_-^3D)+12\mathcal R/A_-^3\,,\\
 h_C&=TY_1/(A_-^2D^2)+2D_1/(A_-^2D^2)
          +6TPD_1/(A_-^3D^2)+6TP/(A_-^3D^2)\,,\\
 h_d&=2/(A_-^2D)+6\mathcal RP/A_-^3\,,\\
 h_Y&=\max\left(\left|A_+^{-2}-a_0^2\right|,
                    \left|\mathcal R A_-^{-2}-a_0^2\right|\right)\,.
 \label{eq:tail-H-derivative-bounds}
\end{align}
The factors $t$ multiplying decaying quantities are controlled by $\sigma T>1$; the rational factors $t/(t+C)$ and $t/(t+C)^2$ are largest at $T$ for $C<0$.
In the last line the lower endpoint includes $t=\infty$:
\begin{equation}
 A_+^{-2}-a_0^2\leq H_Y=\frac{t}{A^2(t+C)}-a_0^2
 \leq\mathcal R A_-^{-2}-a_0^2\,.
 \label{eq:tail-HY-interval}
\end{equation}

Put $k_0=(2a_0)^{-1}$, $\boldsymbol f=(f_A,f_p,f_C,f_d,k_0f_Y,k_0f_Y)$ and $\boldsymbol h=(h_A,h_p,h_C,h_d,k_0h_Y,k_0h_Y)$.
The elementary weighted integral bounds give the nonnegative majorant matrix
\begin{equation}
 \begin{aligned}
 M_{1j}&=f_j/\sigma^2\,,&M_{2j}&=f_j/\sigma\,,\\
 M_{3j}&=k_0(\delta_{j5}+\delta_{j6})/\sigma^2\,,&
 M_{4j}&=k_0(\delta_{j5}+\delta_{j6})/\sigma\,,\\
 M_{5j}&=h_j/(a_0+\sigma)\,,&M_{6j}&=h_j/(a_0-\sigma)\,,
 \end{aligned}
 \label{eq:tail-majorant-matrix}
\end{equation}
and $\norm{D\mathcal T}\leq\max_i q_i^{-1}\sum_jM_{ij}q_j$.
For example, the four scalar kernel bounds are $\sigma^{-1}$, $\sigma^{-2}$, $(a_0+\sigma)^{-1}$, and $(a_0-\sigma)^{-1}$ for the future first integral, future double integral, unstable future integral, and stable past integral, respectively.
Exact rational evaluation and upward rounding give the six row sums
\begin{equation}
 \begin{aligned}
 (&0.581554029033969,\ 0.581554029033969,\ 0.581609542790419\,,\\
 &0.581609542790419,\ 0.581970658254018,\ 0.582081478039319)\,.
 \end{aligned}
 \label{eq:tail-row-sums}
\end{equation}

The residual envelopes give componentwise centered defects bounded by
\begin{equation}
 \left(K_F/a_0^2,\ K_F/a_0,\ 0,\ 0,\
       K_J/(2a_0),\ K_J/[e_-(a_0-\sigma)]\right)\,,
 \label{eq:tail-component-defects}
\end{equation}
where $e_-=2.718281828459045<\sum_{j=0}^{18}1/j!<e$.
The final entry uses $\sup_{\tau\geq0}\tau\e^{-(a_0-\sigma)\tau} =1/[e(a_0-\sigma)]$.
Dividing by the weights proves~\eqref{eq:tail-defect}.
Finally $A\geq A_-$, $t+C\geq D$, and $1-tp/A\geq1-T(s_p+q_p\varepsilon)/A_-$, since $a_0T,\sigma T>1$.
Also $B=1+(1+C)/(t-1)>1$; its global infimum is $1$ at infinity, whereas $(T+C_-)/(T-1)$ bounds $B(T)$ only.

\section{Validated compact core flow}
\label{app:core}

This appendix specifies the enclosure propagated by the core executable.
Write the first order system and parameter coordinates as
\begin{align}
 X'&=\mathcal F(t,X)\,, & X&=(A,p,C,d,Y,J)\,,\nonumber\\
 b&=b_0+10^{-41}\xi_b\,, & R&=R_0+10^{-27}\xi_R\,,\nonumber\\
 |\xi_b|,|\xi_R|&\le r\,, & r&=5\times10^{-4}\,.
 \label{eq:core-parameter-coordinates}
\end{align}
The state at $\xi=0$, its two first variations, and a bound on its three second variations are enclosed separately.

\subsection{Validated horizon initialization}

Put $t_0=20/19$ and denote by $W_j,Q_j$ the $j$th $z$ derivatives evaluated at $z=R/20$.
The exact initialization map is
\begin{align}
 A&=\frac{t_0}{W_0}\,,
 &p&=\frac{1}{W_0}-\frac{R W_1}{t_0W_0^2}\,,
 \nonumber\\
 C&=\frac{Q_0t_0^2(t_0-1)}{R}-t_0\,,
 &d&=\frac{(3t_0^2-2t_0)Q_0}{R}+(t_0-1)Q_1-1\,,
 \nonumber\\
 Y&=\frac{(6t_0-2)Q_0}{R}
       +\left(4-\frac{2}{t_0}\right)Q_1
       +\frac{R(t_0-1)}{t_0^2}Q_2\,,
 \nonumber\\
 J&=\frac{6Q_0}{R}+\frac{6Q_1}{t_0}
       +\frac{3RQ_2}{t_0^2}
       +\frac{R^2(t_0-1)}{t_0^4}Q_3\,.
 \label{eq:core-initial-map}
\end{align}
The horizon polynomials of degree 180 and~\eqref{eq:general-horizon-tail} give the center state enclosure.
Automatic differentiation of the same rational map gives the first and second parameter derivatives.

The $10^{-70}$ padding applied to the truncated first order dual data follows from the following bounds.
For $|z|\le0.035$, the coefficient majorants imply
\begin{equation}
 |W_0|<2,\ |W_1|<16,\ |W_2|<200,\ |W_3|<6000\,,\qquad
 |Q_j|<(8,77,770,7700)_j\,,
 \label{eq:initial-jet-size}
\end{equation}
and both fourth derivatives are smaller than $10^6$.
The full and truncated data lie in these boxes with $W_0\ge1/2$ and $0.69\le R\le0.70$.
Denote~\eqref{eq:core-initial-map} by $G$ and let $y$ comprise its eight jet arguments and $R$.
Direct rational interval differentiation gives
\begin{equation}
 \max_i\sum_{\text{eight jets }u}|G_{i,u}|<10^3\,,\qquad
 \max_i\sum_{u,v}|G_{i,uv}|<10^6\,.
 \label{eq:initial-chain-rule}
\end{equation}
The ancillary script \texttt{horizon\_initialization\_bounds.py} checks these two small rational calculations explicitly.
Cauchy's estimate from Appendix~\ref{app:horizon-remainders} and $z_{\xi_R}=10^{-27}/20$ give $|\partial_{\xi_a}y|_\infty<10^{-20}$.
Equation~\eqref{eq:horizon-common-jet-tail} then bounds the error in $G$ by $10^{-72}$ and the error in either first dual coefficient by
\begin{equation}
 10^3\,10^{-100}+10^6\,10^{-75}\,10^{-20}
 <1.00000001\times10^{-89}\,.
 \label{eq:initial-dual-padding}
\end{equation}
Thus $10^{-70}$ is a conservative value and first derivative padding, not an assumed truncation accuracy.

The second order initialization evaluates the full parameter box, including the six remainder components in~\eqref{eq:horizon-common-jet-tail}.
With the initial positive weights defined below, exact rational postprocessing of that output gives
\begin{equation}
 \begin{aligned}
 \|X_{\xi_b\xi_b}(t_0)\|_q&<2.820\times10^{-78}\,,\\
 \|X_{\xi_b\xi_R}(t_0)\|_q&<5.128\times10^{-67}\,,\\
 \|X_{\xi_R\xi_R}(t_0)\|_q&<2.327\times10^{-52}\,.
 \end{aligned}
 \label{eq:initial-second-variation}
\end{equation}
All three are strictly smaller than the common initial bound $B_{ab}(t_0)=10^{-50}$.
This check uses the same initial weights as the core executable, including a conservative allowance for their binary128 conversion.

\subsection{Center state Picard tube and Taylor step}

At a coarse step starting at $t_k$, let $X_k$ enclose the exact center state and let $h>0$ be its length.
The executable constructs a rectangular tube $Z$ and verifies
\begin{equation}
 X_k+[0,h]\,\mathcal F([t_k,t_k+h],Z)
       \subset\operatorname{int}Z\,.
 \label{eq:core-picard-inclusion}
\end{equation}
It also checks that the denominators $t$, $A$, and $t+C$ stay away from zero.
The vector field is smooth on this box; the inclusion and the standard continuation argument establish existence throughout the step, and local Lipschitz continuity gives uniqueness.

Let $c_j(X,t)$ denote the $j$th Taylor coefficient obtained by substituting a formal time series into $X'=\mathcal F(t,X)$.
For the step of order 56, the Lagrange remainder gives
\begin{equation}
 X(t_k+h)\in
 \sum_{j=0}^{55}c_j(X_k,t_k)h^j+
 c_{56}(Z,[t_k,t_k+h])h^{56}\,.
 \label{eq:core-order56-step}
\end{equation}
Every coefficient operation and endpoint evaluation is an MPFR interval operation at 512 bits.
The same formula at an interior offset encloses the \emph{value} of the center orbit at a sensitivity substep.
Its local Taylor coefficients are then recomputed from the differential equation.
No derivative is taken of an unspecified remainder coefficient.

There are 485 coarse steps and 1590 sensitivity substeps.
The nominal coarse step lengths in successive ranges are
\begin{equation}
\begin{array}{c|ccccc}
t_k&[t_0,1.2)&[1.2,1.5)&[1.5,2)&[2,4)&[4,40]\\ \hline
h_{\rm nom}&0.005&0.01&0.02&0.05&0.1
\end{array}\,.
\label{eq:core-coarse-schedule}
\end{equation}
Their subdivisions use nominal maximum lengths $0.0005,0.002,0.005,0.01,0.025,0.05$ on the successive ranges with endpoints $t_0,1.2,1.5,2,4,10,40$.
The final step is shortened to end at $t=40$.
These numbers specify step choices; interval times enclose the exact chosen rational steps and their endpoints.

\subsection{First variations and weighted error propagation}

Write $S_a(t)=\partial_{\xi_a}X(t,0)$, $a=b,R$, so
\begin{equation}
 S_a'=M(t)S_a\,,\qquad M(t)=D_X\mathcal F(t,X(t,0))\,.
 \label{eq:core-first-variation}
\end{equation}
At the start of a substep the representation is $S_a\in s_a+[-E_aq,E_aq]$, where $s_a$ is a stored binary128 vector, $q_i>0$, and
\begin{equation}
 \|v\|_q:=\max_i\frac{|v_i|}{q_i}\,.
 \label{eq:core-weighted-norm}
\end{equation}
Every weight is an explicitly supplied positive decimal rational.
They may be chosen using an approximate orbit: their accuracy as eigenvectors is not a premise of the proof.
For a verified Jacobian tube $\boldsymbol M$, the executable computes the following outward upper bound on the logarithmic norm:
\begin{equation}
 \mu_q=\max_i\left[
       \overline{\boldsymbol M}_{ii}
       +\sum_{j\ne i}|\boldsymbol M_{ij}|\,\frac{q_j}{q_i}\right]\,.
 \label{eq:core-lognorm}
\end{equation}
Consequently an initial error of weighted size $E_a$ grows by at most $e^{\mu_qh}$ during the substep.
Here $\mu_q>0$: the first row of the variational system is $S_{a,A}'=S_{a,p}$, and its comparison row is $q_p/q_A>0$.

The first variation time series has order $K=12$.
If $M_j$ enclose the Taylor coefficients of the true center Jacobian, the interval coefficients initialized at the represented center obey
\begin{equation}
 V_{a,0}=s_a\,,\qquad
 V_{a,n+1}=\frac{1}{n+1}\sum_{j=0}^{n}M_jV_{a,n-j}\,.
 \label{eq:core-sensitivity-recurrence}
\end{equation}
A candidate tube $\mathcal S_a$ is validated by the strict inclusion
\begin{equation}
 \sum_{j=0}^{11}V_{a,j}[0,h]^j+
 [-q e^{\mu_qh}E_a,q e^{\mu_qh}E_a]+
 [-h^{12}|V_{a,12}^{\rm tube}|,h^{12}|V_{a,12}^{\rm tube}|]
 \subset\operatorname{int}\mathcal S_a\,,
 \label{eq:core-sensitivity-tube}
\end{equation}
with positive outward inflation in the tube iteration.
The remainder coefficient is computed from~\eqref{eq:core-sensitivity-recurrence} using $\mathcal S_a$ and the Jacobian jets over the center state tube.
The polynomial in this inclusion uses the enclosing coefficients $M_j$, not selected midpoint coefficients.

At the substep endpoint, the same enclosing polynomial and its twelfth order remainder give an interval vector $\mathcal V_a(h)$ for the solution initialized at $s_a$.
A predictor using midpoint coefficients may select the new stored center $s_a^+$, because its error is separately bounded:
\begin{equation}
 E_a^+=\tau_q\left[
 e^{\mu_qh}E_a+
 \max_i\frac{\max\{|(\mathcal V_a)_i^- -(s_a^+)_i|,
                |(\mathcal V_a)_i^+ -(s_a^+)_i|\}}{q_i}
 \right]\,,
 \qquad
 \tau_q=\max_i\frac{q_i}{q_i^+}\,.
 \label{eq:core-error-update}
\end{equation}
All differences, maxima, divisions, exponentials, and weight changes in this bound are enclosed outward.
Distances are measured from the \emph{represented} center to both endpoints; half the interval width is not substituted for this radius.
The factor $\tau_q$ accounts explicitly for every change of weights.

\subsection{Second variations and the uniform parameter bootstrap}

Let $T_{ab}=\partial_{\xi_a}\partial_{\xi_b}X$ throughout the parameter box.
Its equation is
\begin{equation}
 T_{ab}'=D_X\mathcal F\,T_{ab}
       +D_X^2\mathcal F[S_a,S_b]\,.
 \label{eq:core-second-variation}
\end{equation}
The uniform bootstrap assumes that the perturbed orbit remains within $\delta_X=10^{-24}$ of the center orbit in each component and that each first variation remains within $\delta_S=10^{-30}$ of its center first variation.
The Hessian is evaluated by interval automatic differentiation on $Z+[-\delta_X,\delta_X]^6$.
If $V_{a,i}=|\mathcal S_{a,i}|+\delta_S$, define
\begin{equation}
 L_{ab}=\max_i\frac{1}{q_i}
 \sum_{u,v}
 \left|\partial_{X_u}\partial_{X_v}\mathcal F_i\right|
 V_{a,u}V_{b,v}\,.
 \label{eq:core-hessian-source}
\end{equation}
Let $\mu_2$ be the analogue of~\eqref{eq:core-lognorm} on that enlarged state box and $g_2=\max\{1,e^{\mu_2h}\}$.
If $\|T_{ab}(t_k,\xi)\|_q\le B_{ab}$, variation of constants gives the uniform substep bound
\begin{equation}
 \widehat B_{ab}=g_2 B_{ab}+h g_2 L_{ab}\,.
 \label{eq:core-second-bound}
\end{equation}
Afterwards $B_{ab}^+=\tau_q\widehat B_{ab}$.
No signed second variation column or unvalidated derivative of an error radius is propagated.

The bootstrap closes by checking, for every component and substep,
\begin{align}
 q_i r(\widehat B_{bb}+\widehat B_{bR})&<\delta_S\,,&
 q_i r(\widehat B_{bR}+\widehat B_{RR})&<\delta_S\,,
 \nonumber\\
 r\bigl(|\mathcal S_{b,i}|+\delta_S+
        |\mathcal S_{R,i}|+\delta_S\bigr)&<\delta_X\,.
 \label{eq:core-bootstrap-inequalities}
\end{align}
The first two inequalities follow from the mean value formula for $S_a(t,\xi)-S_a(t,0)$; the third bounds $X(t,\xi)-X(t,0)$.
Strictness and a first exit argument justify the simultaneous state, first variation, and second variation assumptions.
The maximum ratios of the left sides to the corresponding assumed radii in the regenerated calculation are bounded above by
\begin{equation}
 0.004006962574346467\,,\qquad 0.002383324686807518\,.
 \label{eq:core-bootstrap-rebuilt}
\end{equation}
Both are far below one.

At $t=40$, conservative upper bounds for the two weighted first variation errors are
\begin{equation}
 E_b<1.401659301\times10^{-42}\,,\qquad
 E_R<5.034831333\times10^{-28}\,,
 \label{eq:first-variation-errors}
\end{equation}
and for the weighted Hessian norms,
\begin{equation}
 (B_{bb},B_{bR},B_{RR})
 <(3.794401813,\ 3.794996900,\ 4.227854415)\times10^{-30}\,.
 \label{eq:second-variation-errors}
\end{equation}
The resulting parameter representation is
\begin{equation}
 X_i(40,\xi)\in (X_c)_i+(S_b)_i\xi_b+(S_R)_i\xi_R
                  +[-\mathcal R_i,\mathcal R_i]\,,\qquad
 \mathcal R_i=\frac{q_i r^2}{2}(B_{bb}+2B_{bR}+B_{RR})\,.
 \label{eq:core-parameter-remainder}
\end{equation}
Here $X_c$ and the two $S$ columns are intervals, so their existing state and first variation errors remain included.
The complete endpoint data and the bounds $\mathcal R_i$ are supplied in the ancillary output.

Finally, positivity is checked on each center Picard tube enlarged by $2\delta_X$, which contains the complete parameter family.
It gives the lower bounds~\eqref{eq:A-core-positive}--\eqref{eq:U-core-positive}.
MPFR state endpoints and positivity minima are printed in the appropriate outward direction.
The matching parser also encloses the decimal serialization of the binary128 column centers and upper error bounds before performing exact rational arithmetic.

\section{Poincar\'e--Miranda certificate}
\label{app:poincare}

The matching calculation uses the five coordinates $x=(A,p,C,d,v_0)$, with $v_0=J-a_0Y$, and the exact rational parameter center $\theta_0$ and radii $r_i$ supplied with the certificate.
Write $\delta=\theta-\theta_0$, so that $|\delta_i|\le r_i$.
At $T=40$, the approximate tail is affine in its three parameters:
\begin{equation}
 x_{\mathrm{app}}(\Ai,\beta,\eta)
 =\bigl(\Ai+\eta L_A,\eta L_p,\beta+\eta L_C,
         \eta L_d,\eta\bigr)\,.
 \label{eq:matching-affine-tail}
\end{equation}
The $L$ coefficients are exact rational evaluations of the finite tail shape after its exact normalization $L_J-a_0L_Y=1$.
The correction from this approximation to the exact tail is bounded uniformly over the entire tail parameter box by
\begin{equation}
 r_{\mathrm t}=
 (8.9\times10^{-32},\;5.34\times10^{-32},\;
  1.5225\times10^{-29},\;9.135\times10^{-30},\;5\times10^{-30})\,.
 \label{eq:matching-tail-errors}
\end{equation}
These are the relevant component weights multiplied by the certified tail ball radius; the last entry is a conservative bound for the fixed $\eta$ coordinate.

Let $[c_0]$ enclose the central core endpoint in these five coordinates, and let $[J]$ contain the two core derivative columns and the three exact columns obtained by differentiating the negative of~\eqref{eq:matching-affine-tail}.
Let $r_{\mathrm c}$ denote the core's componentwise second order remainder over the parameter box.
Define
\begin{equation}
 \widehat\Phi_0=\operatorname{mid}[c_0]
                 -x_{\mathrm{app}}(\Aizero,\beta_0,\eta_0)\,,
 \qquad J_0=\operatorname{mid}[J]\,.
 \label{eq:matching-nominal-data}
\end{equation}
The validated core Taylor model and uniform tail ball give the enclosure
\begin{equation}
 \Phi(\theta)=\widehat\Phi_0+J_0\delta+e(\theta)\,,
 \qquad
 |e_i(\theta)|\le\epsilon_i\,,
 \label{eq:matching-affine-enclosure}
\end{equation}
where
\begin{equation}
 \epsilon_i=\operatorname{rad}[c_{0,i}]
 +\sum_{j=1}^{5}\operatorname{rad}[J_{ij}]r_j
 +r_{\mathrm c,i}+r_{\mathrm t,i}\,.
 \label{eq:matching-error-budget}
\end{equation}
Only the first two columns have interval widths.
For this fixed-radius certificate, the tail need only be continuous in its parameters; its complete departure from the affine approximation is already included in $r_{\mathrm t}$.

The matrix $P=J_0^{-1}$ is formed by rational Gaussian elimination, and the identity $PJ_0=I$ is checked exactly.
Put
\begin{equation}
 g_0=P\widehat\Phi_0\,,
 \qquad E_i=\sum_{k=1}^{5}|P_{ik}|\epsilon_k\,.
 \label{eq:matching-preconditioned-budget}
\end{equation}
Thus $G=P\Phi$ satisfies $|G_i(\theta)-g_{0,i}-\delta_i|\le E_i$.
Here $g_0$ is a nominal rational displacement, not the exact value $G(\theta_0)$; the central tail correction is included in $E$.
The face test and its remaining margin are
\begin{equation}
 |g_{0,i}|+E_i<r_i\,,
 \qquad m_i=r_i-|g_{0,i}|-E_i>0\,.
 \label{eq:face-test}
\end{equation}
Consequently $G_i<0$ on $\delta_i=-r_i$ and $G_i>0$ on $\delta_i=r_i$, irrespective of the four transverse coordinates.
The outward bounds are displayed in Table~\ref{tab:poincare-full}.
The smallest normalized margin exceeds $0.5406$, so the Poincar\'e--Miranda theorem supplies a matching zero.

\begin{table}[ht]
\centering
\begin{tabular}{@{}ccccc@{}}
\toprule
$i$ & $r_i$ & $|g_{0,i}|$ & $E_i$ & $m_i/r_i$ \\
\midrule
1 & $5.0\times10^{-45}$ & $2.1025\times10^{-47}$ & $2.2759\times10^{-45}$ & $0.540619770207$ \\
2 & $5.0\times10^{-31}$ & $3.4032\times10^{-33}$ & $2.0935\times10^{-31}$ & $0.574495623641$ \\
3 & $3.0\times10^{-29}$ & $1.9524\times10^{-31}$ & $1.1879\times10^{-29}$ & $0.597534509130$ \\
4 & $2.0\times10^{-27}$ & $1.5037\times10^{-29}$ & $8.9000\times10^{-28}$ & $0.547484807780$ \\
5 & $2.0\times10^{-29}$ & $2.2298\times10^{-32}$ & $7.0264\times10^{-30}$ & $0.647565186013$ \\
\bottomrule
\end{tabular}
\caption{Poincar\'e--Miranda face bounds from the regenerated core.
The radius is exact; the displacement and error are rounded upward, and the normalized margin downward.
The unrounded quantities are rational numbers in the certificate.}
\label{tab:poincare-full}
\end{table}

\paragraph{Transfer of decimal bounds:}
The regenerated MPFR state endpoints are printed with directed rounding.
For every decimal token parsed from the core output, the matching script nevertheless allows one unit $u$ in its last displayed decimal place: a lower endpoint is replaced by $x-u$, an upper endpoint by $x+u$, and a printed center by the interval $[x-u,x+u]$.
Printed nonnegative errors and remainders are increased by their own units before the stated one percent inflation is applied.
The bootstrap ratios are increased and positivity lower bounds decreased in the same way.
This explicitly includes nearest rounding in binary128 center and error serialization.
All subsequent operations, including unit construction, descaling of the parameter columns, inversion, and face testing, are rational.
The auxiliary manuscript data program also rounds every displayed lower bound downward and upper bound upward.

\paragraph{Constraint at both ends:}
The relation between the conserved scalar and the remaining Kundt equation is algebraic.
In dimensionless horizon coordinates, set
\begin{equation}
 \mathcal K=W W_z\mathcal H_z+3W_z^2\mathcal H+W^2
 -\frac{\kappa}{3}
 \left(\mathcal H_z\mathcal H_{zzz}
       -\frac12\mathcal H_{zz}^2+2\right)\,.
 \label{eq:remaining-kundt-residual}
\end{equation}
With $z=R(1-1/t)$, $W=t/A$ and $\mathcal H=-R^2(t+C)/t^3$, direct substitution gives $\mathcal K=-\mathcal N/(6A^4)=\mathcal E_2$.
For the analytic horizon series, $W(0)=1$, $W_z(0)=1+b$, $\mathcal H(0)=0$, $\mathcal H_z(0)=-1$, $\mathcal H_{zz}(0)=4+6b$, and $\mathcal H_{zzz}(0)=-6(1+3b\gamma_2)$.
The prescribed $\gamma_2$ and $\kappa$ therefore give
\begin{equation}
 \mathcal K(0)=b\{-1+\kappa(8+6b-6\gamma_2)\}=0\,.
 \label{eq:horizon-constraint-initialization}
\end{equation}
Analyticity and constraint propagation therefore give $\mathcal N=0$ on the core.
On the tail, $A\to\Ai>0$, $C\to\beta$, and all derivative variables decay exponentially; these rates make every polynomially weighted term in $\mathcal N$ vanish at infinity and give $\mathcal E_2=0$.
The fixed $v_0$ monotonicity test then completes the sixth component as in Section~\ref{sec:proof}.
The supplied symbolic script verifies the coordinate identity, propagation identity, horizon initialization, and fixed $v_0$ derivative with exact algebra.

\paragraph{Physical intervals:}
The same exact parameter box gives the mass and temperature enclosures in Section~\ref{sec:geometry}.
For the literal action~\eqref{eq:action}, the dimensionless entropy has the outward enclosure
\begin{equation}
 \begin{gathered}
 \CertifiedEntropyLower <Gm_2^2 S_{\rm lit}\,,\\
 Gm_2^2 S_{\rm lit}<\CertifiedEntropyUpper\,.
 \end{gathered}
 \label{eq:entropy-interval}
\end{equation}
The last place rounding is included in these endpoints.
The shifted convention is obtained by adding $2\pi$ to $Gm_2^2S_{\rm lit}$, with $\pi$ enclosed by the same rational Machin formula calculation.

\section{Ancillary files and reproducibility}
\label{app:reproducibility}

The ancillary archive contains the complete sources, interval arithmetic header, rational weight schedule, finite tail shape data, build instructions, and output logs used for this revision.
The code in \texttt{src/verified\_core\_taylor\_model.repaired.cpp} is the production core program.
Its interval endpoint serialization is outward, and its sensitivity tubes and recentering errors use the enclosing interval coefficients and distances from the actually stored centers.
The production source is formatted by computational stage, uses descriptive function names, and cites the corresponding equation labels in Appendix~\ref{app:core}.
Unused development helpers have been removed; the deterministic transformation and lexical source consistency check are documented in \path{ancillary/patches/READABILITY.md}.
The sensitivity tube check rejects equality at either boundary, as required by the strict inclusion in~\eqref{eq:core-sensitivity-tube}.
The archive records the source changes and the provenance of the newly constructed interval header and weight schedule.
The archived original core source in \texttt{reference/} is an input to the deterministic source consistency check.
Historical output intervals in that directory are comparison data and do not enter the regenerated match.

\begin{samepage}
The numerical pipeline is run from the ancillary directory by
\begin{verbatim}
make verify
\end{verbatim}
\end{samepage}
This builds and tests the interval arithmetic and binary128 conversion support, regenerates the horizon and compact core enclosures, checks the initial weighted second variation bound, and runs the rational horizon, tail, constraint monotonicity, and matching certificates.
The matching program receives the regenerated core file explicitly.
A failed check terminates the pipeline.
The output logs identify the file actually used and report the bootstrap ratios, positivity bounds, and matching margins.

The reported run used macOS arm64, GNU GCC 15.2.0, MPFR 4.2.2, GMP 6.3.0, GNU Make 3.81, and Python 3.14.7.
The build uses GNU C++17, \texttt{-O2 -fno-fast-math -ffp-contract=off}, MPFR, GMP, and libquadmath.
The supplied binary128 conversion implementation covers the tested platform's missing MPFR conversion entry points.
Its tests include exact round trips, directed conversions of both signs, subnormal boundaries, and ties.
The interval tests compare arithmetic and decimal serialization with exact rational extrema.

The positive component weights are freely chosen rational constants.
The core recomputes its logarithmic norm bounds with those weights and accounts for changes of norm between steps.
Thus the approximate backward orbit and numerical eigenvectors used by \texttt{make weights} only select an efficient schedule.
The proof checks do not assume that this approximate orbit encloses the solution.
A change of schedule requires a new \texttt{make verify} run.

Likewise, the finite tail shape is stored as exact decimal rational data.
Its complete residual and the uniform fixed point estimates are checked by \path{python/tail/centered_tail_uniform_certificate.py}.
The separate \texttt{mpmath} recurrence and shape generation programs are exploratory construction tools; they are not used to justify residual or rounding bounds.
The numerical rational checks use Python's standard library; the full pipeline also uses SymPy for exact algebra.
The latter checks can be run separately as
\begin{verbatim}
python3 python/symbolic/verify_constraints.py
\end{verbatim}
They verify the conserved constraint, its coordinate identification, the horizon constraint, and the fixed $v_0$ scalar derivative; the recorded run used SymPy 1.14.0.
The parameter and thermodynamic intervals and the outward face table are regenerated with
\begin{verbatim}
python3 python/matching/generate_paper_data.py
\end{verbatim}
This program encloses $\pi$ by alternating rational arctangent sums and Machin's identity, so the displayed physical bounds also have a rational rounding certificate.
The entropy bounds refer to the literal action~\eqref{eq:action}, including its Euler density contribution.
The complete package contains every local helper required by the proof.
The public repository containing the sources, data, build instructions, tests, and verification logs is available at \href{\PublicRepositoryURL}{\nolinkurl{\PublicRepositoryURL}}.

\end{document}